\documentclass[conference, 12pt, onecolumn]{ieeeconf}

\IEEEoverridecommandlockouts
\usepackage[usenames]{color}
\usepackage{enumerate}
\usepackage{url}
\usepackage{subfigure}
\usepackage{amsfonts,mathrsfs}
\usepackage{amssymb,amsmath}
\usepackage{verbatim}
\usepackage{acronym}
\usepackage{mathtools}
\usepackage{cite}
\usepackage{graphicx}

\def\fskip#1{}

\newtheorem{theorem}{Theorem}

\newtheorem{definition}{Definition}

\newtheorem{lemma}{Lemma}

\newtheorem{proposition}[theorem]{Proposition}
\newtheorem{remark}{Remark}

\renewcommand{\theenumi}{\arabic{enumi}}

\def\1{{\bf 1}}

\newcommand{\remove}[1]{}

\begin{document}
\title{A Simple Framework for Stability Analysis of State-Dependent Networks of Heterogeneous Agents}
\author{\authorblockN{S. Rasoul Etesami}
  \authorblockA{Department of Industrial and Enterprise Systems Engineering\\ University of Illinois at Urbana-Champaign,  Urbana, IL 61801\\
     Email: etesami1@illinois.edu}
}
\maketitle
\begin{abstract}
Stability and analysis of multi-agent network systems with state-dependent switching typologies have been a fundamental and longstanding challenge in control, social sciences, and many other related fields. These already complex systems become further complicated once one accounts for asymmetry or heterogeneity of the underlying agents/dynamics. Despite extensive progress in analysis of conventional networked decision systems where the network evolution and state dynamics are driven by independent or weakly coupled processes, most of the existing results fail to address multi-agent systems where the network and state dynamics are highly coupled and evolve based on status of heterogeneous agents. Motivated by numerous applications of such dynamics in social sciences, in this paper we provide a new direction toward analysis of dynamic networks of heterogeneous agents under complex time-varying environments. As a result we show how Lyapunov stability and convergence of several challenging problems from opinion dynamics can be established using a simple application of our framework. Moreover, we introduce a new class of asymmetric opinion dynamics, namely nearest neighbor dynamics, and show how our approach can be used to analyze their behavior. In particular, we extend our results to game-theoretic settings and provide new insights toward analysis of complex networked multi-agent systems using exciting field of sequential optimization.
\end{abstract}

\section{Introduction}
Researchers in a number of fields are currently finding a variety of applications for complex networks, and distributed multi-agent network systems are currently the focal point of many new applications. Such applications relate to the growing popularity of online social networks, the analysis of large network data sets, the problems that arise from interactions among agents in complex networks such as formation control, smart grids, political, economic, and biological systems, and the expansion of power and wireless networks in our daily life. 

There is ample evidence that decision making is often guided by heterogeneous agents interacting in a complex time-varying environment. Perhaps one simple example is when a set of heterogeneous robots with different communication capabilities want to rendezvous despite the fact that they are simultaneously moving and yet have to maintain communication connectivity. However, it is often observed that in practice the behavior of multi-agent decision systems under static symmetric/homogeneous setting is fundamentally different from its dynamic asymmetric/heterogeneous counterpart. Unlike the static homogeneous case, often any comprehensive analysis of dynamic heterogeneous multi-agent systems is quite challenging, particularly in dynamic environments, and this class of problems has eluded researchers for many years. New ideas and methodologies need to be developed to address such shortcomings in which any progress can impact numerous applications in variety of domains including opinion formation in social networks, formation control, cyber-physical security, dynamic clustering, among many others. Since better understanding of such complex systems will allow us to design novel or perhaps fundamentally different mechanisms, in this work we take some initial steps toward extending the existing results on multi-agent network systems from the static homogeneous setting to highly dynamic and heterogeneous environments. 

\subsection{Motivation}

There are many motivating examples of relationships in political, social, and engineering applications, which are governed by complex networks of heterogeneous agents. Agents may possibly belong to multiple groups and be connected by multi-layered networks. The networks can also be dynamic in the sense that they can vary over time depending on the agents' status. As a few illustrative examples one can consider (Figure \ref{Fig:min:realization}):

\smallskip
\begin{itemize}
  \item In social networks, there are often clear affinities among people based on shared political or cultural beliefs. However, on specific issues, alliances form among people from different groups. Almost every congressional vote provides an example of this phenomenon, where some representatives break away from their respective parties to vote with the other party. \smallskip
  \item In formation control a basic task is to design a distributed protocol so that a set of robots collectively form a certain structure. Robots have different communication capabilities and can only communicate with those in their local neighborhood. Consequently, the communication network between them may vary depending on their relative distances from each other. \smallskip 
  \item Relationships between countries in the Middle East, and their ties to the US and Russia are nuanced, with affiliations changing over time, and depending on context. Similarly, relationships among terrorist organizations, such as ISIS, Al Qaeda, Taliban, and LeT, often change due to their battle for supremacy, and their fight for the allegiance of their extremist followers. \smallskip
  \item In opinion systems such as political election polls, individuals initially have different opinions about a certain topic/candidate. Individuals frequently interact and become friend/unfriend depending on how close their opinions are from each other. In particular, through such interactions their opinions gradually form and a collective opinion eventually emerges.
\end{itemize}   

\begin{figure}
\vspace{2.75cm}
  \begin{center}
    \includegraphics[totalheight=.16\textheight,
width=.25\textwidth,viewport=875 0 1475 600]{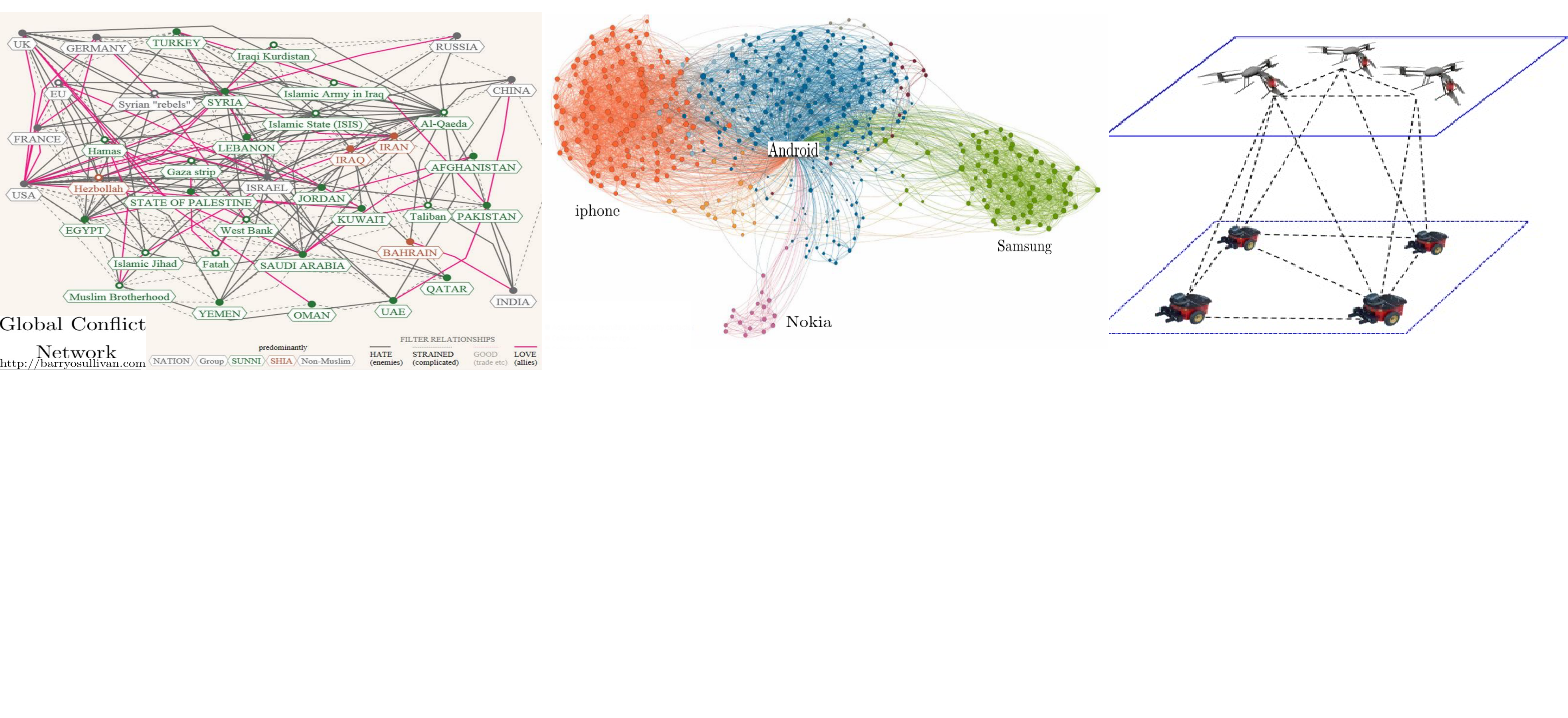}
  \end{center}\vspace{-3cm}
  \caption{\footnotesize{The left figure depicts a complex network of global conflicts between countries (agents) due to their political or religious ties. The middle figure represents a social network of individuals' opinions for buying different cell-phone products. The right figure shows a specific network formation taken by UAVs and drones. In all the figures the networks may dynamically change over time and the agents can be heterogeneous (e.g., drones and UAVs have different communication capabilities).\vspace{-0.2cm}}}\label{Fig:min:realization}
\end{figure} 

\smallskip
Motivated by the above, and many other similar examples, our objective in this work is to provide a simple framework to understand and analyze the behavior of networks of heterogeneous agents with a rich dynamic network structure which may evolve or vary based on agents' states. To this end, we provide new connections between analysis of multi-agent network systems and some developed methods in the mature fields of successive optimization and game theory. Utilizing such connections, we establish Lyapunov stability and convergence of several classes of heterogeneous multi-agent network systems with switching state-dependent dynamics.

\subsection{Literature Review and Organization}
There has been a rich body of literature on analysis of distributed multi-agent network systems, mainly from the static point of view, in which a set of agents iteratively interact over a \emph{fixed} communication network so as to achieve a certain goal such as consensus or optimizing a global objective function. The classical models of Degroot \cite{degroot1974reaching} and Friedkin\& Johnsen \cite{friedkin1997social} in social science are two special types of such systems. As the literature on this area is quite vast, we refer interested readers for an overview to \cite{nedic2018network} and \cite{AVP-RT:17}. However, a crucial assumption in almost all of these works is that the communication network among the agents is fixed. Often these results can be generalized to time-varying networks by assuming a certain ``independency" between the network process and the state dynamics. For instance, one of the commonly used assumptions is that the network dynamics are governed by an exogenous process which is uncoupled from the state dynamics \cite{olshevsky2009convergence,nedic2009distributed,nedic2015distributed,bacsar2016convergence,zhu2011convergence,tatarenko2017non}. A generalization of this idea is to consider multi-agent network dynamics where the network and state dynamics are allowed to be coupled, however a certain condition such as network connectivity or communication symmetry must be satisfied at each time instance \cite{olfati2004consensus,hendrickx2013convergence,jadbabaie2003coordination,touri2012product}. A further extension of these results is the line of works in \cite{sonin2008decomposition,touri2011existence} which shows that any sequence of stochastic matrices viewed as update matrices of multi-agent network dynamics admits a ``Lyapunov" type function. Unfortunately, such Lyapunov functions are based on another so-called \emph{adjoint} dynamics which depend on the future instances of the original dynamics. Thus, unless there is a strong inherent property in the underlying stochastic matrices \cite{etesami2013termination,touri2012product}, it seems unlikely to leverage the co-evolution of original and adjoint dynamics so as to establish meaningful convergence results.

While the above techniques can properly address a large class of multi-agent network systems, there are still many examples which do not fit into any of the aforementioned frameworks (or the application of the above techniques provides very poor results on the behavior of the multi-agent system). One of the main reasons is that most of the developed frameworks for analysis of multi-agent network systems aim to isolate the evolution of \emph{network} dynamics from those of \emph{state} dynamics in which case one derives a conclusion about local time instances (e.g. one time slot or a window of time instances) and then generalizes this behavior to the entire trajectories. But the main question here is that what if the network and state dynamics are completely determined endogenously so that no local property can be assumed or checked a priori? In other words, if we do not take into account the actual correlation between network and state dynamics (or how they evolve in terms of each other), it seems hopeless to have a good understanding of the overall behavior of the dynamics. This shortcoming is even more pronounced once we take into account agents heterogeneity or asymmetric communication among them. In this work we show that despite these challenges it is still possible to capture the co-evolution of network-state dynamics for several (and perhaps many other) multi-agent network dynamics even under asymmetric or heterogeneous environment. While the evolution of state dynamics are commonly captured by difference/differential equations, one of the key novelties of our work is to explicitly derive analogous equations for \emph{network} dynamics by explicitly introducing \emph{network variables} and couple them with the state variables. This allows us to derive network topology from state dynamics and vice versa and even incorporate other constraints into the system behavior.

The rest of the paper is organized as follows: In Section \ref{sec:framework} we provide a general framework which we use frequently throughout the paper. In Section \ref{sec:HK} we apply this framework to a well-known state-dependent switching dynamics from social sciences, known as Hegselmann-Krause (HK) opinion dynamics, and extend it to an edge-heterogeneous setting with additional communication constraints. In particular, we show how this approach can handle up to some extent asymmetric communication among the agents in the HK model. In section \ref{sec:nearest-neighbor} we introduce a new class of asymmetric opinion dynamics, namely nearest neighbor dynamics, and show that our framework can also be applied suitably to these models despite the fact that the underlying communication networks suffer from asymmetry. In Section \ref{sec:game}, we provide an application of our framework to game-theoretic settings and show how the existing results in game theory can be leveraged to analyze heterogeneous multi-agent network dynamics. We conclude the paper by identifying some future directions of research in Section \ref{sec:conclusion}. 

\subsection*{Notations}
We adopt the following notations throughput the paper: We use bold symbols for vectors. Given a vector $\boldsymbol{v}$ we let $diag(\boldsymbol{v})$ be a diagonal matrix with vector $\boldsymbol{v}$ as its diagonal elements and zero, everywhere else. Moreover, we denote the transpose of $\boldsymbol{v}$ by $\boldsymbol{v}^T$. For a positive integer $n\in\mathbb{Z}^+$ we set $[n]:=\{1,2,\ldots,n\}$. We let $\boldsymbol{1}$ be a column vector of all ones. Given a positive-definite matrix $Q$ and a vector $\boldsymbol{v}\in\mathbb{R}^n$, we let $\|\boldsymbol{v}\|^2_Q=\boldsymbol{v}^TQ\boldsymbol{v}$, and $\|\boldsymbol{v}\|$ to be the Euclidean norm of $\boldsymbol{v}$. Given a real number $x\in\mathbb{R}$ we set $(x)^-:=\min\{0,x\}$. We denote the cardinality of a finite set $S$ by $|S|$. 

\section{A Sequential Optimization Framework}\label{sec:framework} 
Often multi-agent dynamical systems which commonly arise in control or social sciences are images of optimization algorithms which are commonly used in machine learning literature. To make this connection more clear, let us first consider the following iterated block successive minimization process which is frequently used in the machine learning literature for minimizing a smooth/non-smooth function \cite{razaviyayn2013unified}. Consider the optimization problem: 
\begin{align}\nonumber
\min f(\boldsymbol{y}_1,\boldsymbol{y}_2,\ldots,\boldsymbol{y}_n), \ \boldsymbol{y}_i\in Y_i, \forall i,
\end{align}
where $Y_i\subseteq \mathbb{R}^{m_i}$ is a closed convex set, and $f: \prod_{i=1}^{n}Y_i\to \mathbb{R}$ is a continuous function. A popular approach to solve the above optimization problem is the \emph{block coordinate descent} (BCD) method. At each iteration of this method, the objective function is minimized with respect to a single block of variables while the rest of the blocks are held \emph{fixed}. More specifically, at iteration $t=0,1,\ldots$ of the algorithm, the block variable $\boldsymbol{y}_i$ is updated by solving the following subproblem:         
\begin{align}\nonumber
\boldsymbol{y}_i^t=\arg\min_{\boldsymbol{z}_i\in Y_i} f(\boldsymbol{y}^{t}_1,\ldots,\boldsymbol{y}_{i-1}^{t},\boldsymbol{z}_i,\boldsymbol{y}_{i+1}^{t},\ldots,\boldsymbol{y}^t_n), \ i\in[n].
\end{align}
In particular, an important question here is whethere or not the generated sequence $\{\boldsymbol{y}^t\}_{t=0}^{\infty}$ where $\boldsymbol{y}^t=(\boldsymbol{y}^t_1,\ldots,\boldsymbol{y}^t_n)$ will converge to a local/global minimizer of the objective function $f(\cdot)$.  Due to its particular simple implementation, the BCD method has been widely used for solving problems such as power allocation in image denoising and image reconstruction, wireless communication systems, clustering, and dynamic programming \cite{howson1975new,hartigan1979algorithm,razaviyayn2013unified}. On the other hand, since in practice finding the exact minimum in each iteration with respect to a block variable might be expensive, one can consider different variants of the BCD method, such as \emph{inexact} BCD method, where one adds a smooth regularizer to the objective function or approximates it by a smooth upper bound function. In either case, and under some mild assumptions, it can be shown that the BCD method will converge to a stationary point of the objective function $f(\cdot)$ \cite{razaviyayn2013unified}. 

To see how BCD method can be used toward stability analysis of multi-agent network systems, let us consider a special case of the above minimization where there are only two block variables, namely a \emph{state} block variable $\boldsymbol{y}:=(\boldsymbol{y}_1,\ldots,\boldsymbol{y}_n)\in \mathbb{R}^{n\times d}$, where $\boldsymbol{y}_i\in\mathbb{R}^{d}$ denotes the state of agent $i$, and a \emph{network} block variable $\boldsymbol{\lambda}:=(\lambda_{ij})\in \Lambda\subseteq [0,1]^{n\times n}$, where $\lambda_{ij}=1$ if there is a directed edge from agent (node) $i$ to agent $j$ so that $\boldsymbol{y}_i$ can be influenced by $\boldsymbol{y}_j$, and $\lambda_{ij}=0$ if no such an edge exists. In other words, an integral block variable $\boldsymbol{\lambda}$ encodes the adjacency matrix of the communication network among agents. Note that we also allow the network to contain self-loops whenever $\lambda_{ii}=1$ for some $i$.

Now let us consider a class of multi-agent network dynamics with $n$ agents evolving over discrete time instances $t=0,1,2,\ldots$ as:
\begin{align}\label{eq:coupled-dynamics}
&\boldsymbol{x}(t+1)= g_1(\boldsymbol{x}(t),\boldsymbol{\lambda}(t)),\cr 
&\boldsymbol{\lambda}(t+1)= g_2(\boldsymbol{x}(t+1)),
\end{align}
where $\boldsymbol{x}(t):=(x_1(t),\ldots,x_n(t))\in X\subseteq \mathbb{R}^{n\times d}$ denotes agent $i$'s state at time $t$, and $\boldsymbol{\lambda}(t)=(\lambda_{ij}(t))_{ij}\in \Lambda\subseteq [0,1]^{n\times n}$ identifies the communication links between each pair of agents at that time. Here $g_1(\cdot)$ and $g_2(\cdot)$ are two functions which capture the update rule of the underlying dynamics. As it can be seen from \eqref{eq:coupled-dynamics}, the state of the dynamics at the next time instance $\boldsymbol{x}(t+1)$ is determined by the joint pair of state-network at time $t$, i.e., $(\boldsymbol{x}(t),\boldsymbol{\lambda}(t))$, while the network structure at time $t+1$ is determined by the state variable at time $t$.

\begin{proposition}\label{prop:framework}
Let $(S,\leq_S)$ be a totally ordered set.\footnote{In most applications of this paper (with an exception of Section \ref{sec:nearest-neighbor}) we set $(S,\leq_S)$ to be the set of real numbers endowed by its natural order.} Moreover, assume that there exists a function $\Phi(\boldsymbol{y},\boldsymbol{\lambda}):X\times \Lambda\to S$, such that given any fixed network $\boldsymbol{\lambda}^*\in \Lambda$,  $\Phi(\boldsymbol{y},\boldsymbol{\lambda}^*):X\to S$ is nonincreasing (with respect to $\leq_S$) along the image of $g_1(\cdot)$, i.e.,
\begin{align}\label{eq:state}
\Phi\big(g_1(\boldsymbol{y},\boldsymbol{\lambda}^*),\boldsymbol{\lambda}^*\big)\leq_S \Phi\big(\boldsymbol{y},\boldsymbol{\lambda}^*\big), \ \forall \boldsymbol{y}\in X.  
\end{align}
If there exists a function $f_1:\Lambda\to S$ such that for any fixed state $\boldsymbol{y}\in X$,
\begin{align}\label{eq:network}
g_2(\boldsymbol{y})\in \arg\min_{\boldsymbol{\lambda}\in \Lambda} \{f(\boldsymbol{y},\boldsymbol{\lambda}):=\Phi(\boldsymbol{y},\boldsymbol{\lambda})+f_1(\boldsymbol{\lambda})\},
\end{align}
then $f(\boldsymbol{y},\boldsymbol{\lambda})$ is nonincreasing (with respect to $\leq_S$) along the trajectories of \eqref{eq:coupled-dynamics}.
\end{proposition}
\begin{proof}
Using the definition of joint dynamics \eqref{eq:coupled-dynamics} we can write:
\begin{align}\nonumber
f(\boldsymbol{x}(t+1),\boldsymbol{\lambda}(t+1))&=f(\boldsymbol{x}(t+1),g_2(\boldsymbol{x}(t+1)))\cr 
&=\min_{\boldsymbol{\lambda}\in \Lambda} f(\boldsymbol{x}(t+1),\boldsymbol{\lambda})\cr 
&\leq_S f(\boldsymbol{x}(t+1),\boldsymbol{\lambda}(t))\cr 
&=\Phi\big(g_1(\boldsymbol{x}(t),\boldsymbol{\lambda}(t)),\boldsymbol{\lambda}(t)\big)+f_1(\boldsymbol{\lambda}(t))\cr 
&\leq_S\Phi\big(\boldsymbol{x}(t),\boldsymbol{\lambda}(t)\big)+f_1(\boldsymbol{\lambda}(t))\cr 
&=f(\boldsymbol{x}(t),\boldsymbol{\lambda}(t)),
\end{align}   
where the second equality is due to \eqref{eq:network}, and the last inequality holds by \eqref{eq:state} given the fixed $\boldsymbol{\lambda}^*=\boldsymbol{\lambda}(t)$. As a result, the coupled dynamics in \eqref{eq:coupled-dynamics} can be replicated by applying the BCD method to the objective function $f(\boldsymbol{y},\boldsymbol{\lambda})=\Phi(\boldsymbol{y},\boldsymbol{\lambda})+f_1(\boldsymbol{\lambda})$ with constraint sets $\boldsymbol{\lambda}\in \Lambda$ and $\boldsymbol{x}\in X$, where at each iteration we fix either network or state variable and optimize the objective function $f$ with respect to the other variable.
\end{proof}

\begin{definition}
Let $(S,\leq_S)$ be a totally ordered set. A function $V:\mathbb{R}^{m}\to S$ is called a Lyapunov function for the discrete time dynamical system $\boldsymbol{z}(t+1)=h(\boldsymbol{z}(t))$, if it is nonincreasing along the trajectories of the dynamics, i.e., $V(\boldsymbol{z}(t+1))\leq_S V(\boldsymbol{z}(t))$. We refer to a dynamical system which admits a Lyapunov function as Lyapunov stable.
\end{definition}

Intuitively, Proposition \ref{prop:framework} states that if trajectories of the projected state dynamics in \eqref{eq:coupled-dynamics} over a \emph{fixed} network $\boldsymbol{\lambda}^*$ admit a Lyapunov function $\Phi(\boldsymbol{y},\boldsymbol{\lambda}^*)$, while minimizing $f(\boldsymbol{y},\boldsymbol{\lambda})=\Phi(\boldsymbol{y},\boldsymbol{\lambda})+f_1(\boldsymbol{\lambda})$ with respect to $\boldsymbol{\lambda}\in \Lambda$ accurately captures the network associated to the fixed state $\boldsymbol{y}$, then the joint state-network dynamics \eqref{eq:coupled-dynamics} admit a Lyapunov function. In particular, $V(\boldsymbol{y})=\min_{\boldsymbol{\lambda}\in \Lambda}f(\boldsymbol{y},\boldsymbol{\lambda})$ serves as a Lyapunov function for the dynamics $\{\boldsymbol{x}(t), t=0,1,\ldots\}$ generated by \eqref{eq:coupled-dynamics}. In what follows next we show how this simple framework can be used to establish Lyapunov stability and convergence of several important state-dependent multi-agent network dynamics.

\section{Hegselmann-Krause Opinion Dynamics}\label{sec:HK}

To show effectiveness of the proposed framework in Section \ref{sec:framework} toward stability and convergence analysis of multi-agent network systems, in this section we consider a well-known model from opinion dynamics known as \emph{Hegselmann-Krause} (HK) model \cite{hegselmann2002opinion}. A natural question that commonly arises in social sciences is the extent to which one can predict the outcome of the opinion formation of entities under some complex interaction process running among these social actors \cite{degroot1974reaching,hegselmann2002opinion,friedkin1999social,friedkin2011social,lorenz2007continuous,proskurnikov2018tutorial}. In this regard, one of the first studies was undertaken by Hegselmann and Krause in \cite{hegselmann2002opinion} with many applications in the robotics rendezvous \cite{Bullo,chazelle2011total}, linguistic formation \cite{dong2016dynamics}, social networks \cite{ye2018evolution}, trust and marketing \cite{wai2015identifying}, among many others \cite{proskurnikov2018tutorial}. In the HK model, a finite number of agents frequently update their opinions where the opinion of each agent is captured by a scalar (or vector) quantity in one (or higher) dimension.\footnote{For simplicity of presentation, in this section we only consider one dimensional HK model. However, all the results can be extended in a straightforward manner to higher dimensions.} Because of the conservative nature of social entities, each agent in this model communicates only with those whose opinions are closer to him and lie within a certain level of his confidence. 
 
In the \emph{homogeneous} HK model, there are a set of $[n]$ agents. It is assumed that at each time instance $t=0, 1, 2, \ldots$, the opinion (state) of agent $i\in[n]$ can be represented by a scalar $x_{i}(t)\in \mathbb{R}$. Each agent $i$ updates its value at time $t$ by taking the arithmetic average 
of its own value and those of all the others that are in its \emph{$\epsilon$-neighborhood} at time $t$. Here the parameter $\epsilon>0$ is a constant which captures the confidence bound. More precisely, the evolution of opinion vectors $\boldsymbol{x}(t):=(x_1(t),\ldots,x_n(t))\in\mathbb{R}^n$ can be modeled by the following discrete-time dynamics:
\begin{align}\label{eqn:DynamicForm}
&\boldsymbol{x}(t+1)=A(t)\boldsymbol{x}(t),\cr
&A_{ij}(t) = \begin{cases} \frac{1}{|N_i(t)|} & \mbox{if } j\in N_i(t), \\
 0 & \mbox{ else}, \end{cases}
\end{align}
where  $N_i(t)$ is the set of neighbors of agent $i$, i.e., 
\begin{align}\nonumber
N_i(t)=\{j\in[n]: |x_i(t)-x_j(t)|\leq \epsilon\}.
\end{align}

In the \emph{node heterogeneous} HK dynamics everything remains as above except that different agents can have different confidence bounds $\epsilon_i,  i\in [n]$. The node heterogeneous model reflects the fact that some agents are very open minded (large $\epsilon_i$) and are willing to communicate with many others before updating their opinions, while some agents are closed-minded (small $\epsilon_i$) and are biased towards their own opinions. For instance $\epsilon_i=0$ means that agent $i$ is stubborn and will not change its opinion at all. Although at first glance the differences between homogeneous and node-heterogeneous HK dynamics may seem negligible, their outcomes are substantially different, such that most of the results from one cannot be carried over to the other \cite{lorenzt,mirtabatabaei2012opinion,lorenz2010heterogeneous,julient}. In this regard, one of the fundamental questions concerning HK dynamics is whether or not they eventually converge to a final outcome.

\subsection{Homogeneous HK Model}
Let us first focus on the \emph{homogeneous} HK model. Although convergence and detailed analysis of the homogeneous HK model have been established and studied extensively in the past literature (see, e.g., \cite{proskurnikov2018tutorial} for a comprehensive survey), in this subsection we provide a simple argument to show why this model fits into our framework. This will allow us to generalize stability of homogeneous HK model to account for higher degrees of heterogeneity or asymmetry among the agents. 

Let us define $\mathcal{G}_t=([n],\mathcal{E}_t)$ to be the communication graph at a generic time $t$ such that there is an edge between agents $i$ and $j$ at time $t$, i.e., $(i,j)\in\mathcal{E}_t$ if and only if $j\in N_i(t)$. Note that in the homogeneous HK model the communication graph is undirected as if agent $j$ is a neighbor of agent $i$, the converse is also true. In fact, what makes the analysis of HK dynamics challenging is the strong coupling between the evolution of the network $\mathcal{G}_t$ and the state $\boldsymbol{x}(t)$. This is because at any time $t$ the state vector $\boldsymbol{x}(t)$ determines the network topology $\mathcal{G}_t$, and this new network determines the state vector at the next time step $\boldsymbol{x}(t+1)$. This puts HK dynamics to the class of complex time-dependent and state-dependent network dynamics \cite{chazelle2017inertial,pineda2013noisy,blondel20072r,hendrickx2013symmetric,bhattacharyya2013convergence,roozbehani2008lyapunov}. In particular, the communication network $\mathcal{G}_t$ may switch many times depending on how the opinion vectors evolve which brings additional complication to the analysis.

Now let us consider the following objective function comprised of two block variables, namely $\boldsymbol{y}\in \mathbb{R}^n$ and $\boldsymbol{\lambda}=(\lambda_{ij})\in [0,1]^{n\times n}$,
\begin{align}\label{eq:HK}
f(\boldsymbol{y},\boldsymbol{\lambda}):=\sum_{i,j}\lambda_{ij}\Big((y_i-y_j)^2-\epsilon^2\Big).
\end{align} 
This function can also be written in a compact form as $f(\boldsymbol{y},\boldsymbol{\lambda})=\boldsymbol{y}^T\mathcal{L}\boldsymbol{y}-tr(\mathcal{L})\epsilon^2$, where $\mathcal{L}:=diag(\boldsymbol{\lambda}\boldsymbol{1})-\boldsymbol{\lambda}$, and $tr(\cdot)$ denotes the trace function. Intuitively, the block variable $\boldsymbol{\lambda}$ is meant to capture the communication network $\mathcal{G}_t$, and the block variable $\boldsymbol{y}$ captures the opinion states. Note that if we restrict $\lambda_{ij}$s to binary variables in $\{0,1\}$, then $\boldsymbol{\lambda}$ simply represents the adjacency matrix of a network of $n$ agents where $\lambda_{ij}=1$ if there is a directed edge $(i,j)$ from node $i$ to node $j$, and $\lambda_{ij}=0$ otherwise. Moreover, for such a binary block variable $\boldsymbol{\lambda}$, the matrix $\mathcal{L}$ is precisely the \emph{Laplacian} matrix of the communication network associated with $\boldsymbol{\lambda}$. Although we still need to assume that $\boldsymbol{\lambda}\in\{0,1\}^{n\times n}$, to avoid complication of handling integral variables, for now we allow $\lambda_{ij}$s to vary continuously in the interval $[0,1]$. As we shall see soon the integrality of network variables will be automatically achieved during iterations of the BCD method.

Now let us consider the BCD method applied to the objective function \eqref{eq:HK} with block variables $\boldsymbol{x}$ and $\boldsymbol{\lambda}$. For a generic time $t$, let us fix the state variable to $\boldsymbol{y}=\boldsymbol{x}(t)$. Minimizing \eqref{eq:HK} with respect to the network variable $\boldsymbol{\lambda}\in\Lambda=[0,1]^{n\times n}$ we obtain,
\begin{align}\label{eq:lamba}
\boldsymbol{\lambda}_t:=\arg\min_{\boldsymbol{\lambda}\in[0,1]^{n^2}} \sum_{i,j}\lambda_{ij}\Big((x_i(t)-x_j(t))^2-\epsilon^2\Big),
\end{align}
where 
\begin{align}\label{eq:homogeneous-HK-lambda}
(\boldsymbol{\lambda}_t)_{ij}=(\boldsymbol{\lambda}_t)_{ji} = \begin{cases} 1 & \mbox{if } |x_i(t)-x_j(t)|\leq \epsilon, \\
 0 & \mbox{else}. \end{cases}
\end{align}
This simply follows because the objective function $f(\boldsymbol{x}(t),\boldsymbol{\lambda})$ is a linear function of the network variable $\boldsymbol{\lambda}$ and achieves its minimum in an extreme point of $[0,1]^{n\times n}$. In particular, the optimal extreme point can be found by an easy inspection as given in \eqref{eq:homogeneous-HK-lambda}. But note that $\boldsymbol{\lambda}_t$ is precisely the adjacency matrix of the communication graph $\mathcal{G}_t$ in the homogeneous HK model (recall that in the homogeneous HK model two agents are each others' neighbors at time $t$ if and only if their distance is at most $\epsilon$). Therefore, fixing the state variable to $\boldsymbol{x}(t)$ and minimizing \eqref{eq:HK} with respect to the network variable exactly delivers the adjacency matrix of the communication network $\mathcal{G}_t$ in the homogeneous HK model at that time. 

Now let us fix the network variable to the minimizer $\boldsymbol{\lambda}_t$ which represents an undirected graph with associated Laplacian matrix $\mathcal{L}_t$. It is well-known that given a fixed undirected graph with Laplacian matrix $\mathcal{L}_t$ and arbitrary values $\{y_i, i\in[n]\}$ at its $n$ verticies, the quadratic function $\Phi(\boldsymbol{y},\boldsymbol{\lambda}_t):=\boldsymbol{y}^T\mathcal{L}_t\boldsymbol{y}$ is nonincreasing if each node updates its value to the average value of its neighbors \cite{zhang2012lyapunov,alex}. In other words, defining $\boldsymbol{y}'=A_t \boldsymbol{y}$ where $A_t:=(diag(\boldsymbol{\lambda}_t\boldsymbol{1}))^{-1}\boldsymbol{\lambda}_t$, we have $(\boldsymbol{y}')^T\mathcal{L}_t\boldsymbol{y}'\leq \boldsymbol{y}\mathcal{L}_t\boldsymbol{y}$. But note that $A_t$ is precisely the update matrix of the homogeneous HK model given in \eqref{eqn:DynamicForm}, which means that for the specific choice of $\boldsymbol{y}=\boldsymbol{x}(t)$ we have $\boldsymbol{y}'=\boldsymbol{x}(t+1)$. Thus the evolution of the homogeneous HK dynamics is governed by the application of BCD method to the objective function \eqref{eq:HK}. Appealing to Proposition \ref{prop:framework}, this shows that $f(\boldsymbol{x}(t+1),\boldsymbol{\lambda}(t+1))\leq f(\boldsymbol{x}(t),\boldsymbol{\lambda}(t))$, implying that homogeneous HK dynamics admit a Lyapunov function.  

\begin{remark}
Adapting the same notation as in Section \ref{sec:framework} we have 
\begin{align}\nonumber
&g_1(\boldsymbol{y},\boldsymbol{\lambda}_t)=A_t\boldsymbol{y}, \ \ \ \ \ \ \ \ \ g_2(\boldsymbol{y})=(\boldsymbol{1}_{\{|y_i-y_j|\leq\epsilon\}})_{ij},\cr 
&\Phi(\boldsymbol{y},\boldsymbol{\lambda}):=\boldsymbol{y}^T\mathcal{L}\boldsymbol{y}, \ \ \ \ \ \ \ f_1(\boldsymbol{\lambda})=-tr(\mathcal{L})\epsilon^2, 
\end{align}
and $f(\boldsymbol{y},\boldsymbol{\lambda})=\boldsymbol{y}^T\mathcal{L}\boldsymbol{y}-tr(\mathcal{L})\epsilon^2$ so that the first term $\boldsymbol{y}^T\mathcal{L}\boldsymbol{y}$ captures the internal coupling between network and opinion states in the HK model.    
\end{remark}

\subsection{Restricted Edge Heterogeneous HK Dynamics}

In this part we show how the framework of Section \ref{sec:framework} can extend the analysis of homogeneous HK model by capturing higher degrees of heterogeneity or constraints. For this purpose, we consider \emph{restricted edge-heterogeneous} HK model which is a variant of the homogeneous HK model with the following two additional changes: 

\smallskip
{\bf 1) Edge Heterogeneity}: Consider the same dynamical system as in the homogeneous HK model \eqref{eqn:DynamicForm}, except that the distance between every pair of nodes is measured based on possibly a different confidence bound. More precisely, let $\{\epsilon_{ij}=\epsilon_{ji}>0, \forall i\neq j\}$ be a set of fixed thresholds (one for each pair of agents) so that agents $i$ and $j$ at time instance $t$ can communicate if and only if their distance at that time is less than $\epsilon_{ij}$, and thus $N_i(t)=\{j\in [n]: |x_i(t)-x_j(t)|\leq \epsilon_{ij}\}$. This captures the heterogeneous relationship among individuals due to family or other social ties. For instance two family members will still continue to communicate even if their opinions are relatively far from each other, while two strangers are more likely to terminate their interactions as soon as their opinions slightly deviate from each other. As before we assume that at each time instance agents update their opinions by taking the arithmetic average of their neighbors' opinions, determined based on the heterogeneous thresholds $\epsilon_{ij}$. Note that the homogeneous HK dynamics is a special case of the edge heterogeneous setting where $\epsilon_{ij}=\epsilon, \forall i\neq j$.

\smallskip
{\bf 2) Communication Restrictions}: Other than closeness in opinion, often there are other important factors which determine whether or not two agents should communicate. For instance, individuals often communicate with those who have closer opinion to them \emph{and} are within small geographic distance from them. One direct way of handling such restriction is to add an extra component to each agent's opinion where this new component encodes the geographic location of that agent. As a result, in this higher dimensional opinion space two agents are each others' neighbors if each component of their opinion vectors (and in particular the geographic component of their opinion vector) are close to each other. This implies that two agents are eligible to communicate if they are close to each other both opinionwise and geographic-wise.\footnote{Although this approach increases the dimension of the opinion space, yet most of the results such as convergence of the dynamics can be extended to this higher dimensional setting \cite{etesami2015game}.} An alternative approach however for imposing new constraints is to consider a predefined underlying network $\mathcal{G}$ which restricts the agents' interactions to only those who are both connected through the edges of $\mathcal{G}$ and have closer opinion to each other. Intuitively, one can imagine running HK dynamics over the graph $\mathcal{G}$. Thus, denoting the communication network of the edge-heterogeneous HK dynamics at time $t$ by $\mathcal{G}(t)$, the actual communication graph at time $t$ is given by the intersection of edges which appear in both $\mathcal{G}$ and $\mathcal{G}(t)$. This second approach is particularly more suitable when there are certain \emph{hard} communication constraints among individuals due to age gaps, gender restrictions, or other social laws.

\begin{theorem}\label{thm:edge-heterogeneous}
The restricted edge-heterogeneous HK model over an undirected graph $\mathcal{G}=([n],\mathcal{E})$ is Lyapunov stable. Moreover, in the absence of edge heterogeneity (i.e., when all pairs have the same confidence bound), the restricted HK model converges to an equilibrium point geometrically fast.
\end{theorem}
\begin{proof}
Given a set of pairwise thresholds $\{\epsilon_{ij}=\epsilon_{ji}>0: i,j\in [n]\}$, and an undirected restricting graph $\mathcal{G}=([n],\mathcal{E})$, let us consider the BCD method applied to the following minimization problem:
\begin{align}\nonumber
&\min f(\boldsymbol{y},\boldsymbol{\lambda}):=\sum_{i,j}\lambda_{ij}\Big((y_i-y_j)^2-\epsilon_{ij}^2\Big)\cr 
& \ \ \ \ \ \ \ \ \boldsymbol{\lambda}\in \Lambda, \ \ \boldsymbol{y}\in \mathbb{R}^n,
\end{align} 
where $\Lambda=\{(\lambda_{ij})\in [0,1]^{n\times n}:  \lambda_{ij}=0, \forall \{i,j\}\notin \mathcal{E}, \lambda_{ij}=\lambda_{ji} \ \forall i,j\}$ is the constraint set for the block variable $\boldsymbol{\lambda}$. The above minimization problem can be rewritten as
\begin{align}\label{eq:restricted}
\min &\sum_{\{i,j\}\in \mathcal{E}}\lambda_{ij}\Big((y_i-y_j)^2-\epsilon_{ij}^2\Big),\cr 
& \ \lambda_{ij}\in[0,1],  \forall \{i,j\}\in \mathcal{E}, \ \ \boldsymbol{y}\in \mathbb{R}^n.
\end{align}
Now given a fixed state variable $\boldsymbol{y}=\boldsymbol{x}(t)$ at a generic time $t$, minimizing \eqref{eq:restricted} with respect to block variable $\boldsymbol{\lambda}$ gives us $\boldsymbol{\lambda}_t$ where
\begin{align}\nonumber
(\boldsymbol{\lambda}_t)_{ij}=(\boldsymbol{\lambda}_t)_{ji} = \begin{cases} 1 & \mbox{if } |x_i(t)-x_j(t)|\leq \epsilon_{ij}, \{i,j\}\in\mathcal{E}  \\
 0 & \mbox{else}. \end{cases}
\end{align}
In other words, minimizing \eqref{eq:restricted} with respect to $\boldsymbol{\lambda}\in \Lambda$ precisely captures the communication structure in the restricted edge-heterogeneous HK model. Now let us fix the network variable to $\boldsymbol{\lambda}_t$. Since $\boldsymbol{\lambda}_t$ represents the adjacency matrix of an undirected graph with corresponding Laplacian $\mathcal{L}_t$, as before $\Phi(\boldsymbol{y},\boldsymbol{\lambda}_t):=\boldsymbol{y}^T\mathcal{L}_t\boldsymbol{y}$ is nonincreasing for the HK update rule over this fixed network $\boldsymbol{\lambda}_t$. Thus for any $\boldsymbol{y}\in\mathbb{R}^n$, if we denote the update matrix of the restricted edge-heterogeneous HK model on the undirected graph $\boldsymbol{\lambda}_t$ by $A_t$, we have $(A_t\boldsymbol{y})^T\mathcal{L}_t(A_t\boldsymbol{y}) \leq \boldsymbol{y}^T\mathcal{L}\boldsymbol{y}$. In particular, by choosing $\boldsymbol{y}=\boldsymbol{x}(t)$, where $\boldsymbol{x}(t)$ denotes the state vector of the restricted edge-heterogeneous HK model at time $t$, we have  $\boldsymbol{x}(t+1)^T\mathcal{L}_t\boldsymbol{x}(t+1) \leq \boldsymbol{x}^T(t)\mathcal{L}_t\boldsymbol{x}(t)$. Therefore, BCD method applied to \eqref{eq:restricted} replicates the dynamics of the restricted edge-heterogeneous HK model. In particular, this shows that $V(\boldsymbol{y}):=\min_{\boldsymbol{\lambda}\in \Lambda}f(\boldsymbol{y},\boldsymbol{\lambda})=\sum_{\{i,j\}\in \mathcal{E}}\Big((y_i-y_j)^2-\epsilon_{ij}^2\Big)^-$
serves as a Lyapunov function for the restricted edge-heterogeneous HK model. 
 
 
In what follows next, we use the above Lyapunov function to establish asymptotic convergence of the restricted HK model to an equilibrium point in the absence of edge heterogeneity (i.e., when all $\epsilon_{ij}$ are the same which by rescaling from now we may assume $\epsilon_{ij}=1,\forall i,j$). To lower bound the decrease of Lyapunov function $V(\cdot)$ at a given time step $t$, we note that this decrease is lower bounded by the decrease amount which is achieved due to the state update. Thus   
\begin{align}\label{eq:hk-drift-restricted}
V(\boldsymbol{x}(t))-V(\boldsymbol{x}(t+1))&\ge \boldsymbol{x}^T(t)\mathcal{L}_t\boldsymbol{x}(t)-\boldsymbol{x}(t+1)^T\mathcal{L}_t\boldsymbol{x}(t+1)\cr 
&=\boldsymbol{x}^T(t)(\mathcal{L}_t-A_t^T\mathcal{L}_tA_t)\boldsymbol{x}(t)\cr 
&=\boldsymbol{x}^T(t)(I-A_t)^T(diag(\boldsymbol{\lambda_t}\boldsymbol{1})+\boldsymbol{\lambda_t})(I-A_t)\boldsymbol{x}(t)\cr
&=(\boldsymbol{x}(t)-\boldsymbol{x}(t+1))^T(diag(\boldsymbol{\lambda_t}\boldsymbol{1})+\boldsymbol{\lambda_t})(\boldsymbol{x}(t)-\boldsymbol{x}(t+1))\cr 
&\ge \|\boldsymbol{x}(t)-\boldsymbol{x}(t+1)\|^2,
\end{align} 
where the first equality is obtained by using $\boldsymbol{x}(t+1)=A_t\boldsymbol{x}(t)$, and the second equality is valid by a simple matrix multiplication and noting that $diag(\boldsymbol{\lambda}_t\boldsymbol{1})A_t=\boldsymbol{\lambda}_t$. Finally the last inequality holds because $\lambda_t$ is the adjacency matrix of a connected undirected graph,\footnote{Here, without loss of generality we may assume $\lambda_t$ to be connected, otherwise for the rest of analysis we can restrict our attention to one of its connected components.} and hence $diag(\boldsymbol{\lambda_t}\boldsymbol{1})+\boldsymbol{\lambda_t}$ is a positive definite matrix whose eigenvalues are greater than or equal to 1. By summing \eqref{eq:hk-drift-restricted} for all $\tau\leq t$, and rearranging the terms we get  
\begin{align}\nonumber
\sum_{\tau=0}^{t}\|\boldsymbol{x}(\tau)-\boldsymbol{x}(\tau+1)\|^2\leq V(\boldsymbol{x}(0))-V(\boldsymbol{x}(t))\leq V(\boldsymbol{x}(0))+n^2,
\end{align}
where the second inequality is valid since by the definition of $V(\cdot)$ we always have $V(\cdot)\ge -n^2$. Thus $\sum_{\tau=0}^{\infty}\|\boldsymbol{x}(\tau)-\boldsymbol{x}(\tau+1)\|^2=\sum_{\tau=0}^{\infty}\sum_{i=1}^{n}(x_i(\tau)-x_i(\tau+1))^2$ is a convergent series, and hence for any $\delta^2>0$, there exists a sufficiently large time $t_{\delta}$ such that $\sum_{\tau=t_{\delta}}^{\infty}\sum_{i=1}^{n}(x_i(\tau)-x_i(\tau+1))^2<\delta^2$.  On the other hand, it is shown in Lemma \ref{lemm:eulerian-graph} that if $\delta<\frac{1}{n^2}$, no switch in the communication network can occur after time $t_{\delta}$. Thus after at most finite time $t_{\delta}$ the communication network of the restricted HK model remains unchanged. This implies that from time $t_{\delta}$ onward the evolution of the dynamics is governed by powers of a fixed stochastic matrix which is well-known to converge to an equilibrium point geometrically fast.
\end{proof}

\begin{remark}
The update matrix $A_t$ of the HK model is the transition matrix of a lazy \emph{simple} random walk on its underlying network $\boldsymbol{\lambda}_t$. However, $\Phi(\boldsymbol{y},\boldsymbol{\lambda}_t):=\boldsymbol{y}^T\mathcal{L}_t\boldsymbol{y}$ serves as a Lyapunov function for any irreducible random walk on the undirected graph $\boldsymbol{\lambda}_t$. As a result, one can allow more general update weights $\{w_{ij}>0, i,j\in[n]\}$ than original weights $\{\frac{1}{k}, k\in[n]\}$ appearing in the update matrices of the HK model, and still use the above analysis to show that the generated dynamics are Lyapunov stable. This can be done by replacing variables $\lambda_{ij}$ by $w_{ij}\lambda_{ij}$ in the above proof.      
\end{remark}
 
\subsection{Asymmetric 0-1 HK Dynamics}
Finally, in this subsection we take one step further and consider a special case of the \emph{node-heterogeneous} HK model. As we mentioned earlier, the dynamics of node-heterogeneous HK model follow exactly the same update rule as homogeneous HK model given in \eqref{eqn:DynamicForm} except that different agents might have different confidence bounds $\epsilon_i, i\in [n]$. Unfortunately, up to the time of writing this paper there is no general result which either proves or disproves convergence of the node-heterogeneous HK model (although some partial results concerning stability of these dynamics are known \cite{su2017partial,etesami2015game,mirtabatabaei2012opinion}). In particular, in the recent work \cite{chazelle2017inertial}, the authors have used an algorithmic approach to show convergence of a special case of the node-heterogeneous HK model, namely 0-1 HK model, in which the confidence bound of each node is restricted to be either 0 or $\epsilon=1$ (i.e., $\epsilon_i\in \{0,1\}, \forall i\in[n]$). It is worth noting that due to heterogeneous confidence bounds, the communication network in the 0-1 HK model is no longer undirected (symmetric). While the proof in \cite{chazelle2017inertial} is fairly long and algorithmic, here we provide a simple argument based on the BCD framework to establish Lyapunov stability of the 0-1 HK dynamics. An important advantage of our approach is that i) it provides an improved Lyapunov drift which can be useful towards convergence rate analysis, and ii) it provides a clear explanation of why the asymmetric 0-1 HK dynamics can still be treated as the symmetric homogeneous HK model. 

\smallskip
\begin{theorem}\label{thm:0-1}
The 0-1 HK dynamics are Lyapunov stable. 
\end{theorem}
\begin{proof}
Let us consider the same function $f(\boldsymbol{y},\boldsymbol{\lambda})$ as in \eqref{eq:HK}. We show that this function is nonincreasing over the trajectories of the 0-1 HK. Let $S_0$ and $S_1=[n]\setminus S_0$ denote the set of agents with confidence bounds $0$ and $\epsilon=1$, respectively. As before, given a fixed state variable $\boldsymbol{y}:=\boldsymbol{x}(t)$, minimizing $f$ with respect to $\boldsymbol{\lambda}\in \Lambda:=[0,1]^{n\times n}$ we obtain $(\boldsymbol{\lambda}_t)_{ij}=(\boldsymbol{\lambda}_t)_{ji}=1$ if $|x_i(t)-x_j(t)|\leq \epsilon$, and $(\boldsymbol{\lambda}_t)_{ij}=(\boldsymbol{\lambda}_t)_{ji}=0$, otherwise. This correctly captures the communication links adjacent to the agents in $S_1$ of the actual communication network in the 0-1 HK model at time $t$. However, it is possible that $\boldsymbol{\lambda}_t$ incorrectly sets $(\boldsymbol{\lambda}_t)_{ij}=1$ for an agent $i\in S_0$ so that $\boldsymbol{\lambda}_t$ and the actual 0-1 HK communication network at time $t$ can only deviate from each other on edges $\{(i,j), i\in S_0\}$. Nevertheless, as far as it concerns the agents in $S_0$, this will not be an issue since the agents in $S_0$ will never use their adjacent links to update their states (these agents are always fixed). Therefore, we are only left to show that for fixed undirected graph $\boldsymbol{\lambda}_t$ with corresponding Laplacian $\mathcal{L}_t$, $\Phi(\boldsymbol{y},\boldsymbol{\lambda}_t)=\boldsymbol{y}^T\mathcal{L}_t\boldsymbol{y}$ still serves as a Lyapunov function for the 0-1 HK update rule.

Given an arbitrary state vector $\boldsymbol{y}$, let us decompose it into $\boldsymbol{y}=(\boldsymbol{y}_0,\boldsymbol{y}_1)$, where $\boldsymbol{y}_0$ and $\boldsymbol{y}_1$ are associated to the agents in $S_0$ and $S_1$, respectively. Define $R_0:=\boldsymbol{\lambda}_t[S_0]$ and $R_1:=\boldsymbol{\lambda}_t[S_1]$ to be the adjacency matrices induced by $\boldsymbol{\lambda}_t$ over the agents in $S_0$ and $S_1$, respectively. Moreover, let $M:=\boldsymbol{\lambda}_t[S_1,S_0]$ be the adjacency matrix of the bipartite graph induced by $\boldsymbol{\lambda}_t$ between agents in $S_0$ and $S_1$. Therefore, we can write 
\begin{align}\nonumber
\boldsymbol{\lambda}_t=
\left[
\begin{array}{c|c}
R_0 & M^T \\
\hline
M & R_1
\end{array}
\right], \ \ \mathcal{L}_t=
\left[
\begin{array}{c|c}
D_0-R_0 & -M^T \\
\hline
-M & D_1-R_1
\end{array}\right],
\end{align} 
where $D_0=diag(R_0\boldsymbol{1})$ and $D_1=diag(R_1\boldsymbol{1})$. Moreover, the actual 0-1 HK update rule at time $t$ can be written as $\boldsymbol{x}(t+1)=A_t\boldsymbol{x}(t)$, where
\begin{align}\nonumber
A_t=
\left[
\begin{array}{c|c}
I & \boldsymbol{0} \\
\hline
D_1^{-1}M & D_1^{-1}R_1
\end{array}
\right].
\end{align}
Since $A_t\boldsymbol{y}=(\boldsymbol{y}_0,D_1^{-1}M\boldsymbol{y}_0+D_1^{-1}R_1\boldsymbol{y}_1)$, we can write
\begin{align}\nonumber
\boldsymbol{y}^T\mathcal{L}_t\boldsymbol{y}&=\boldsymbol{y}^T_0(D_0\!-\!R_0)\boldsymbol{y}_0+\boldsymbol{y}^T_1(D_1\!-\!R_1)\boldsymbol{y}_1\!-\!2\boldsymbol{y}^T_1 M\boldsymbol{y}_0,\cr 
\boldsymbol{y}^T\!A_t^T\mathcal{L}_tA_t\boldsymbol{y}\!&=\!\boldsymbol{y}^T_0(\!D_0\!-\!R_0\!)\boldsymbol{y}_0\!-\!2(\!D_1^{-1}\!M\boldsymbol{y}_0\!+\!D_1^{-1}\!R\boldsymbol{y}_1\!)^T\!M\boldsymbol{y}_0\cr 
&\qquad\!+\!(D_1^{-1}M\boldsymbol{y}_0\!+\!D_1^{-1}R\boldsymbol{y}_1)^T(D_1\!-\!R_1)(D_1^{-1}M\boldsymbol{y}_0\!+\!D_1^{-1}R\boldsymbol{y}_1).
\end{align} 
Subtracting these two expressions from each other, using the fact that $R_1^T=R_1$, and simplifying the terms we obtain
\begin{align}\label{eq:0-1hk-right}
\boldsymbol{y}^T\mathcal{L}_t\boldsymbol{y}-\boldsymbol{y}^TA_t^T\mathcal{L}_tA_t\boldsymbol{y}&=\boldsymbol{y}^T_0\Big(2M^TD_1^{-1}M-M^TD_1^{-1}(D_1-R_1)D_1^{-1}M\Big)\boldsymbol{y}_0\cr 
&\qquad+\boldsymbol{y}^T_1\Big(D_1-R_1-R_1D_1^{-1}(D_1-R_1)D_1^{-1}R_1\Big)\boldsymbol{y}_1\cr 
&\qquad+2\boldsymbol{y}^T_1\Big(R_1D_1^{-1}R_1D_1^{-1}M-M\Big)\boldsymbol{y}_0.
\end{align} 
Now a straightforward calculation shows that the right-hand side of \eqref{eq:0-1hk-right} can be factorized as $P^T(D_1+R_1)P$, where 
$P:=(I-D_1^{-1}R_1)\boldsymbol{y}_1-D_1^{-1}M\boldsymbol{y}_0$. Finally, since $\boldsymbol{y}-A_t\boldsymbol{y}=(\boldsymbol{0},(I-D_1^{-1}R_1)\boldsymbol{y}_1-D_1^{-1}M\boldsymbol{y}_0)$, we can rewrite \eqref{eq:0-1hk-right} as
 \begin{align}\label{eq:last-drift}
\boldsymbol{y}^T\!\mathcal{L}_t\boldsymbol{y}\!-\!(A_t\boldsymbol{y})^T\!\mathcal{L}_t(A_t\boldsymbol{y})\!=\!(\boldsymbol{y}\!-\!A_t\boldsymbol{y})^T\!\!\left[\!\!
\begin{array}{cc}
I & \!\!\!\boldsymbol{0} \\
\boldsymbol{0} &  \!\!\! D_1\!+\!R_1
\end{array}
\!\!\!\!\right]\!\!(\boldsymbol{y}\!-\!A_t\boldsymbol{y}).
\end{align} 
Therefore, if we define $Q$ to be the middle matrix in \eqref{eq:last-drift}, $Q$ would be a positive definite matrix (as its diagonal elements are strictly positive and dominate the row-sums) and thus $\boldsymbol{y}^T\mathcal{L}_t\boldsymbol{y}-(A_t\boldsymbol{y})^T\!\mathcal{L}_t(A_t\boldsymbol{y})=\|\boldsymbol{y}-A_t\boldsymbol{y}\|^2_Q\ge 0$.  Finally, choosing $\boldsymbol{y}=\boldsymbol{x}(t)$ we get, 
\begin{align}\nonumber
\boldsymbol{x}^T\!(t)\mathcal{L}_t\boldsymbol{x}(t)\!-\!\boldsymbol{x}^T\!(t\!+\!1)\mathcal{L}_t\boldsymbol{x}(t\!+\!1)\!=\!\|\boldsymbol{x}_1(t)\!-\!\boldsymbol{x}_1(t\!+\!1)\|^2_{Q}.
\end{align} 
Therefore, $V(\boldsymbol{y}):=\min_{\boldsymbol{\lambda}\in[0,1]^{n^2}}f(\boldsymbol{y},\boldsymbol{\lambda})=\sum_{i,j}\big((y_i-y_j)^2-\epsilon^2\big)^-$ serves as a Lyapunov function for the 0-1 HK model so that $V(\boldsymbol{x}(t))-V(\boldsymbol{x}(t+1))\ge \|\boldsymbol{x}_1(t)\!-\!\boldsymbol{x}_1(t\!+\!1)\|^2_{Q}$.           
\end{proof}

\begin{remark}
One can view the update matrix $A_t$ of the 0-1 HK model at a given time $t$ as the transition matrix of a lazy simple random walk with absorbing states $S_0$ on the fixed \emph{actual} communication graph at time $t$. Therefore, in the second part of the proof of Theorem \ref{thm:0-1} we have shown that although the actual graph might have one-sided directed edges from $S_1$ to the absorbing states $S_0$, $\boldsymbol{y}^T\mathcal{L}_t\boldsymbol{y}$ still serves as a Lyapunov function for such absorbing random walks where $\mathcal{L}_t$ is the Laplacian of the \emph{symmetrized} actual network $\boldsymbol{\lambda}_t$ (i.e., viewing one-sided edges as undirected edges).  
\end{remark}

As we close this section, we would like to mention that unlike homogeneous HK model which is known to converge to an equilibrium point after finitely many steps \cite{hegselmann2002opinion}, it may take arbitrarily long time until the 0-1 HK dynamics converge. In fact, it seems impossible to show that the drift of the above Lyapunov function is bounded below by a time-invariant quantity (as is the case for homogeneous HK model \cite{etesami2015game}). As an example, consider a set of $n=3$ agents initially positioned at $x_1(0)=-1+\frac{1}{2^{2^m}}$, $x_2(0)=0$, and $x_3(0)=1$ where $m$ can be any arbitrary large integer. Also assume $\epsilon_1=\epsilon_2=0$, and $\epsilon_3=1$ so that agent $3$ is the only moving agent. Then it takes $2^m$ steps until a switch in the communication network occurs so that agent $3$ be able to observe agent $1$. In particular, the drift of the above Lyapunov function during iteration $t\in \{1,\ldots,2^m\}$ is $\frac{1}{2^t}-\frac{1}{2^{t+1}}=\frac{1}{2^{t+1}}$ which can be arbitrarily small.

\section{Nearest Neighbor Opinion Dynamics}\label{sec:nearest-neighbor}
In general, loosing symmetry in communication networks of multi-agent systems can substantially complicate their stability analysis which in turn requires novel techniques. In fact, unlike the symmetric case, existing results concerning stability of state-dependent networks of multi-agent systems with asymmetric communication typologies are quite limited. Nevertheless, this shall not eliminate the possibility of convergence of asymmetric dynamics to an equilibrium point, as it is shown in this section for a special class of \emph{nearest neighbor} dynamics. More specifically, in this section we establish convergence of a class of  nearest neighbor dynamics under both \emph{asynchronous} and \emph{synchronous} settings, where in the former at each time instance only one of the agents updates its opinion (state), while in the latter at each iteration all the agents simultaneously update their opinions.

\subsection{Asynchronous Nearest Neighbor Dynamics}
Consider a set of $[n]$ agents where the opinion of agent $i\in[n]$ at time $t=0,1,2,\ldots$ is given by a vector $x_i(t)\in \mathbb{R}^d$. At each iteration $t$ one agent $i\in[n]$ is selected based on some selection rule (e.g. uniformly at random) and updates its opinion at the next time step to $x_i(t+1)=\mu_i x_i(t)+(1-\mu_i)x_{r(i)}(t)$, where here $r(i):=\arg\min_{j\in [n]\setminus\{i\}}\|x_i(t)-x_j(t)\|$ denotes the closest agent to $i$ with respect to profile $\boldsymbol{x}(t)\in \mathbb{R}^{n\times d}$, and $\mu_i\in(0,1)$ is an agent-specific parameter. For all other agents $j\neq i$, we set $x_j(t+1)=x_j(t)$. 


The rationale behind introducing nearest neighbor dynamics is that often individuals get influenced by their closest friend/partner/leader so that depending on their stubbornness (captured by $\mu_i$) they are willing to compromise in order to get closer to their friend/partner/leader. It is important to note that the communication network in the nearest neighbor dynamics is asymmetric so that if $r(i)$ is the closest agent to $i$, it does not imply that $i$ is also the closest agent to $r(i)$ (i.e., in general $r(r(i))\neq i$).\footnote{In fact, one can show that the communication network at each time instance is comprised of disjoint directed trees where the out-degree of each node is equal to 1.} Moreover, the communication network which determines the ``closest relationships" evolves as a function of agents' opinions which can switch many times based on trajectory of the dynamics. Nevertheless, as we shall see in Theorem \ref{thm:asynchronous-nearest neighbor} such heterogeneous asymmetric state-dependent dynamics will converge to an \emph{equilibrium} point as defined below: 

\smallskip
\begin{definition}
An $\epsilon$-equilibrium for the nearest neighbor dynamics is an opinion profile where the maximum distance between every agent's opinion and its closest neighbor is at most $\epsilon>0$. Moreover, given an initial opinion profile $\boldsymbol{x}(0)$, we let $t_\epsilon$ be the first time instance when $\boldsymbol{x}(t_{\epsilon})$ becomes an $\epsilon$-equilibrium.
\end{definition}

It is worth noting that the definition of $\epsilon$-equilibrium implies that each agent will lie within a distance of at most $\epsilon n$ from its limit point. This is because if the communication network in an $\epsilon$-equilibrium is connected, the maximum distance between every two nodes is at most $\epsilon n$ so that the convex hull of all the opinions can have diameter at most $\epsilon n$. Since the nearest neighbor dynamics evolve inside of this convex hull for all the future iterations, the limit point (if it exists) will also lie in this convex hull. Similarly, if an $\epsilon$-equilibrium contains more than one connected component, then either the distance between their convex hulls is less than $\epsilon$, in which case a similar argument as above for the convex hull of the union of those components can be applied, or the distance between those components is more than $\epsilon$, in which case those components evolve separately from each other so that the limit points of each component lie within its own convex hull (and again the above argument applies).

\smallskip
\begin{definition}
Given two vectors $\boldsymbol{u},\boldsymbol{v}\in \mathbb{R}^n$, we say $\boldsymbol{u}$ is lexicographically smaller than $\boldsymbol{v}$ ($\boldsymbol{u}<_{Lex}\boldsymbol{v}$), if there exists some $k\leq n$ for which $u_1=v_1,\ldots, u_{k-1}=v_{k-1}$, and $u_{k}<v_{k}$. Note that there is no specific relation between components of $\boldsymbol{u}$ and $\boldsymbol{v}$ for indices larger than $k$.  
\end{definition}

\smallskip
\begin{theorem}\label{thm:asynchronous-nearest neighbor}
The asynchronous nearest neighbor dynamics are Lyapunov stable and asymptotically converge to an equilibrium point. Moreover, if at each iteration an agent is selected uniformly at random to update its opinion, then the expected number of steps until the dynamics reach an $\epsilon$-equilibrium is bounded above by $\mathbb{E}[t_\epsilon]\leq \frac{n2^nD_0}{(1-\mu_{\max})\epsilon}$, where $D_0=\max_{i,j} \|x_i(0)-x_{j}(0)\|$ and $\mu_{\max}=\max_i\mu_i$.
\end{theorem}
\begin{proof}
Given a vector $\boldsymbol{v}$, let $sort(\boldsymbol{v})$ be a vector obtained by sorting all the components of $\boldsymbol{v}$ in a nondecreasing order. Let $f:[0,1]^{n^2}\times\mathbb{R}^{n\times d}\to\mathbb{R}^{n}$ be the vector function $f(\boldsymbol{\lambda},\boldsymbol{y}):=sort\Big(\sum_{j=1}^n\lambda_{ij}\|y_i-y_j\|, i\in[n]\Big)$, and consider the \emph{lexicographical} minimization problem, $\min_{Lex} \{f(\boldsymbol{\lambda},\boldsymbol{y}): \boldsymbol{y}\in \mathbb{R}^{n\times d}, \boldsymbol{\lambda}\in \Lambda\}$, where $\Lambda=\{(\lambda_{ij})\in [0,1]^{n^2}: \sum_{j=1}^{n}\lambda_{ij}=1,\ \lambda_{ii}=0, \forall i\in [n]\}$.\footnote{In terms of Proposition \ref{prop:framework} terminology, here we have $(S,\leq_S)=(\mathbb{R}^n,<_{Lex})$.} Now given a fixed block variable $\boldsymbol{y}=\boldsymbol{x}(t)$ at a generic time $t$, minimizing $f(\boldsymbol{\lambda},\boldsymbol{x}(t))$ lexicographically with respect to $\boldsymbol{\lambda}\in \Lambda$ gives us,
\begin{align}\label{eq:match-network}
(\boldsymbol{\lambda}_t)_{ij}=\begin{cases}
1 & \mbox{if } j=r(i), \\ 
0 & \mbox{else}.
\end{cases} 
\end{align}
This is because minimizing $f(\boldsymbol{\lambda},\boldsymbol{x}(t))$ over $\boldsymbol{\lambda}\in \Lambda$ decomposes into minimizing $f(\cdot)$ componentwise, and it is achieved by setting the coefficient $\lambda_{ij}$ corresponding to the smallest term $\|x_i(t)-x_{r(i)}(t)\|$ of the $i$th component equal to 1, and to 0, otherwise. As a result, given a fixed state $\boldsymbol{y}=\boldsymbol{x}(t)$, minimizing $f(\boldsymbol{\lambda},\boldsymbol{x}(t))$ with respect to $\boldsymbol{\lambda}\in S$ accurately captures the directed communication network $\boldsymbol{\lambda}_t$ of the nearest neighbor dynamics for the state $\boldsymbol{x}(t)$. Next let us fix the communication network to $\boldsymbol{\lambda}_t$ so that
\begin{align}\nonumber
f(\boldsymbol{\lambda}_t,\boldsymbol{x}(t))&=sort\Big(\|x_i(t)-x_{r(i)}(t)\|, i\in[n]\Big)\cr 
&=\Big(\|x_1(t)-x_{r(1)}(t)\|,\ldots,\|x_n(t)-x_{r(n)}(t)\|\Big),
\end{align}where in the second equality and without loss of generality (by relabeling the agents if necessary) we have assumed that $\|x_1(t)-x_{r(1)}(t)\|\leq\ldots\leq \|x_n(t)-x_{r(n)}(t)\|$. 

To study the effect of state update on $f(\boldsymbol{\lambda}_t,\boldsymbol{x}(t))$, let us assume that at time $t$ agent $\ell\in[n]$ is selected to update its opinion. Then we obtain $\boldsymbol{x}(t+1)$ for which $x_{\ell}(t+1)=\mu_{\ell}x_{\ell}(t)+(1-\mu_{\ell})x_{r(\ell)}(t)$, and $x_j(t+1)=x_j(t), \forall j\neq \ell$. In particular,
\begin{align}\label{eq:lth-agent}
\|x_{\ell}(t\!+\!1)-x_{r'(\ell)}(t\!+\!1)\|&\leq \|x_{\ell}(t\!+\!1)-x_{r(\ell)}(t\!+\!1)\|\cr 
&=\|x_{\ell}(t\!+\!1)-x_{r(\ell)}(t)\|=\mu_{\ell}\|x_{\ell}(t)-x_{r(\ell)}(t)\|\cr 
&<\|x_{\ell}(t)-x_{r(\ell)}(t)\|,
\end{align}              
where $r'(\ell)$ denotes the closest agent to $\ell$ with respect to the opinion profile $\boldsymbol{x}(t+1)$. Furthermore, for every $i<\ell$ we have two possibilities: Case I) $r(i)\neq \ell$, in which case
\begin{align}\nonumber
\|x_{i}(t\!+\!1)\!-\!x_{r'(i)}(t\!+\!1)\|&=\min\{\|x_{i}(t)\!-\!x_{r(i)}(t)\|,\|x_{i}(t)\!-\!x_{\ell}(t\!+\!1)\|\}\leq\|x_{i}(t)\!-\!x_{r(i)}(t)\|.
\end{align}
Case II) $r(i)=\ell$, in which case by definition of $r(\ell)$ we must have $\|x_{\ell}(t)-x_{r(\ell)}(t)\|\leq \|x_{\ell}(t)-x_{i}(t)\|=\|x_{r(i)}(t)-x_{i}(t)\|$. However, as $i<\ell$ (and thus $\|x_{r(i)}(t)-x_{i}(t)\|\leq \|x_{\ell}(t)-x_{r(\ell)}(t)\|$), one can see that Case II cannot happen unless $\|x_{\ell}(t)-x_{r(\ell)}(t)\|=\|x_{i}(t)-x_{r(i)}(t)\|$. As a result $r(\ell)=i$, and we can write
\begin{align}\nonumber
\|x_{i}(t\!+\!1)-x_{r'(i)}(t\!+\!1)\|&=\|x_i(t)-x_{\ell}(t\!+\!1)\|\cr 
&=\|x_i(t)-\mu_{\ell}x_{\ell}(t)-(1-\mu_{\ell})x_{r(\ell)}(t)\|\cr 
&=\|x_i(t)-\mu_{\ell}x_{r(i)}(t)-(1-\mu_{\ell})x_{i}(t)\|\cr 
&=\mu_{\ell} \|x_{i}(t)-x_{r(i)}(t)\|<\|x_{i}(t)-x_{r(i)}(t)\|.
\end{align}  

Therefore, we have shown that $\|x_{i}(t\!+\!1)-x_{r'(i)}(t\!+\!1)\|\leq \|x_{i}(t)-x_{r(i)}(t)\|, \forall i<\ell$, and moreover $\|x_{\ell}(t\!+\!1)-x_{r'(\ell)}(t\!+\!1)\|<\|x_{\ell}(t)-x_{r(\ell)}(t)\|$. In other words, after agent $\ell$'s update, $f(\boldsymbol{\lambda}_t,\boldsymbol{x}(t))$ decreases lexicographically, i.e., $f(\boldsymbol{\lambda}_{t+1},\boldsymbol{x}(t\!+\!1))\!<_{Lex}\!f(\boldsymbol{\lambda}_t,\boldsymbol{x}(t))$. In particular, $V(\boldsymbol{y}):=\min_{\boldsymbol{\lambda}\in \Lambda}f(\boldsymbol{\lambda},\boldsymbol{y})= sort\Big(\|y_k-y_{r(k)}\|, k\in[n]\Big)$ serves as a Lyapunov function for the asynchronous nearest neighbor dynamics. 

To show convergence of the asynchronous dynamics to an equilibrium point, let us convert $V(\boldsymbol{y})$ into a scalar function $\hat{V}(\boldsymbol{y}):=\sum_{i=1}^{n} \min_{i}(\{\|y_k-y_{r(k)}\|, k\in[n]\})2^{n-i}$ by giving appropriate weights to its coordinates, where here $\min_i(S)$ denotes the $i$th smallest element of a finite set $S$.\footnote{This conversion encodes lexicographical decrease of $V(\boldsymbol{y})$ into a scalar decrease in $\hat{V}(\boldsymbol{y})$. This will allow us to quantify a convergence rate for the asynchronous nearest neighbor dynamics.} As before, let us assume that at time $t$ agent $\ell$ is selected to update its opinion, and $\min_i(\{\|x_k(t)-x_{r(k)}(t)\|, k\in[n]\})=\|x_i(t)-x_{r(i)}(t)\|$. For any $i\leq \ell$, we have 
\begin{align}\label{eq:ai-bi-positive}
&\min_i\Big(\{\|x_k(t+1)-x_{r'(k)}(t+1)\|, k\in[n]\}\Big)\cr 
&\qquad\qquad\leq \max\{\|x_k(t+1)-x_{r'(k)}(t+1)\|, k\in [i]\}\cr
&\qquad\qquad\leq \max\{\|x_k(t)-x_{r'(k)}(t)\|, k\in [i]\}=\|x_i(t)-x_{r(i)}(t)\|,
\end{align}
where the last inequality is by $\|x_{k}(t+1)-x_{r'(k)}(t+1)\|\leq \|x_{k}(t)-x_{r(k)}(t)\|, \forall k\leq \ell$. Now we can write,
\begin{align}\label{eq:i<l}
&\sum_{i\leq \ell}\Big(\min_i\big(\{\|x_k(t)-x_{r(k)}(t)\|, k\in[n]\}\big)-\min_i\big(\{\|x_k(t+1)\!-\!x_{r'(k)}(t+1)\|, i\in[n]\}\big)\Big)2^{n-i}\cr
&\qquad=\sum_{i\leq \ell}\Big(\|x_i(t)-x_{r(i)}(t)\|-\min_i\big(\{\|x_k(t+1)\!-\!x_{r'(k)}(t+1)\|, k\in[n]\}\big)\Big)2^{n-i}\cr 
&\qquad\ge2^{n-\ell}\sum_{i\leq \ell}\Big(\|x_i(t)-x_{r(i)}(t)\|-\min_i(\{\|x_k(t+1)\!-\!x_{r'(k)}(t+1)\|, k\in[n]\})\Big)\cr
&\qquad=2^{n-\ell}\Big(\sum_{i\leq \ell}\|x_i(t)-x_{r(i)}(t)\|-\sum_{i\leq \ell}\min_i(\{\|x_k(t+1)\!-\!x_{r'(k)}(t+1)\|, k\in[n]\})\Big)\cr
&\qquad\ge 2^{n-\ell}\Big(\sum_{i\leq \ell}\|x_i(t)-x_{r(i)}(t)\|-\sum_{i\leq \ell}\|x_i(t+1)\!-\!x_{r'(i)}(t+1)\|\Big)\cr 
&\qquad=2^{n-\ell}\sum_{i\leq \ell}\Big(\|x_i(t)-x_{r(i)}(t)\|-\|x_i(t+1)\!-\!x_{r'(i)}(t+1)\|\Big)\cr 
&\qquad\ge 2^{n-\ell}(1-\mu_{\ell})\|x_{\ell}(t)-x_{r(\ell)}(t)\|,
\end{align}
where the first inequality is due \eqref{eq:ai-bi-positive}, and the second inequality holds because sum of $\ell$ smallest elements of a set is always smaller than sum of any $\ell$ elements in that set. Finally, the last inequality is valid since the first $\ell-1$ summands are nonnegative and by \eqref{eq:lth-agent} the $\ell$th summand is greater than or equal to $(1-\mu_{\ell})\|x_{\ell}(t)-x_{r(\ell)}(t)\|$.   

On the other hand, we note that for any agent $k$, $\|x_k(t+1)\!-\!x_{r'(k)}(t+1)\|\leq \|x_k(t)\!-\!x_{r'(k)}(t)\|+(1-\mu_{\ell})\|x_k(t)\!-\!x_{r'(k)}(t)\|$. This is because the amount of movement of agent $\ell$ at time $t$ is equal to $\|x_{\ell}(t+1)-x_{\ell}(t)\|=(1-\mu_{\ell})\|x_{\ell}(t)-x_{r(\ell)}(t)\|$. Therefore, by triangle inequality the distance of any agent $k$ to its nearest neighbor at the next time step can increase by at most $(1-\mu_{\ell})\|x_{\ell}(t)-x_{r(\ell)}(t)\|$. Now similar as in \eqref{eq:ai-bi-positive} we can write,   
\begin{align}\label{eq:ai-bi-negative}
&\min_i\Big(\{\|x_k(t+1)-x_{r'(k)}(t+1)\|, k\in[n]\}\Big)\cr 
&\qquad\qquad\leq \max\{\|x_k(t+1)-x_{r'(k)}(t+1)\|, k\in [i]\}\cr 
&\qquad\qquad\leq \max\{\|x_k(t)-x_{r'(k)}(t)\|+(1-\mu_{\ell})\|x_{\ell}(t)-x_{r'(\ell)}(t)\|, k\in [i]\}\cr 
&\qquad\qquad=\|x_i(t)-x_{r(i)}(t)\|+(1-\mu_{\ell})\|x_{\ell}(t)-x_{r'(\ell)}(t)\|.
\end{align}
As a result we obtain, 
\begin{align}\label{eq:i>l}
&\sum_{i>\ell}\Big(\min_i(\{\|x_k(t)-x_{r(k)}(t)\|, k\in[n]\})-\min_i(\{\|x_k(t+1)\!-\!x_{r'(k)}(t+1)\|, i\in[n]\})\Big)2^{n-i}\cr
&\qquad=\sum_{i> \ell}\Big(\|x_i(t)-x_{r(i)}(t)\|-\min_i(\{\|x_k(t\!+\!1)\!-\!x_{r'(k)}(t+1)\|, k\in[n]\})\Big)2^{n-i}\cr 
&\qquad\ge-\sum_{i>\ell}(1-\mu_{\ell})\|x_{\ell}(t)-x_{r'(\ell)}(t)\|2^{n-i}\cr 
&\qquad=-(1-\mu_{\ell})\|x_{\ell}(t)-x_{r'(\ell)}(t)\|(2^{n-\ell}-1),
\end{align}
where the inequality is due to \eqref{eq:ai-bi-negative}. Finally, summing \eqref{eq:i>l} and \eqref{eq:i<l}, and using the definition of $\hat{V}(\cdot)$, we obtain $\hat{V}(\boldsymbol{x}(t))-\hat{V}(\boldsymbol{x}(t+1))\ge (1-\mu_{\ell})\|x_{\ell}(t)\!-\!x_{r'(\ell)}(t)\|$. As agent $\ell$ is selected uniformly at random, this implies that as long as $\boldsymbol{x}(t)$ is not an $\epsilon$-equilibrium, the expected decrease of $\hat{V}(\cdot)$ at time $t$ is at least, 
\begin{align}\label{eq:drift_asynchronous}
\mathbb{E}[\hat{V}(\boldsymbol{x}(t))-\hat{V}(\boldsymbol{x}(t+1))]\ge \frac{1}{n}\sum_{\ell=1}^{n}(1-\mu_{\ell})\|x_{\ell}(t)\!-\!x_{r'(\ell)}(t)\|\ge \frac{(1-\mu_{\max})\epsilon}{n}.
\end{align}
Finally, we note that $\hat{V}(\cdot)$ is a nonnegative function such that $\hat{V}(\boldsymbol{x}(0))\leq 2^nD_0$. This in view of \eqref{eq:drift_asynchronous} shows that the expected number of steps before the asynchronous dynamics reach an $\epsilon$-equilibrium is bounded above by $\frac{2^nD_0}{\frac{(1-\mu_{\max})\epsilon}{n}}$.        
\end{proof}

\subsection{Synchronous Nearest Neighbor Dynamics}
In this part we consider the synchronous nearest neighbor dynamics whose formal definition is as follows: Consider a set of $[n]$ agents where the opinion of agent $i\in[n]$ at time $t=0,1,2,\ldots$ is given by $x_i(t)\in \mathbb{R}^d$. Given the current opinion profile $\boldsymbol{x}(t)$, in the next iteration \emph{every} agent $i\in[n]$ updates its opinion to $x_i(t+1)=\mu_i x_i(t)+(1-\mu_i)x_{r(i)}(t)$, where as before $r(i):=\arg\min_{j\in [n]\setminus\{i\}}\|x_i(t)-x_j(t)\|$ denotes the closest agent to $i$ with respect to the opinion profile $\boldsymbol{x}(t)$, and $\mu_i\in(0,1), i\in[n]$ are mixing constants. As before, we note that the communication network of the synchronous dynamics is asymmetric whose evolution depends on the opinion profiles.

 In the following we show that if all the agents have the same mixing parameter $\mu_i=\mu, \forall i\in [n]$, then the \emph{synchronous} nearest neighbor dynamics reach an $\epsilon$-equilibrium very fast. 
\begin{theorem}
The synchronous nearest neighbor dynamics with $\mu_i:=\mu\in(0,1)$, $\forall i\in[n]$ converge to an $\epsilon$-equilibrium point after at most $t_\epsilon\leq n(\frac{2D_0}{\epsilon}+\log_{|1-2\mu|} (\frac{\epsilon}{2D_0}))$ iterations, where $D_0=\max_{i,j} \|x_i(0)-x_{j}(0)\|$.
\end{theorem}
\begin{proof}
Consider BCD method for \emph{component-wise} minimizing of the function: 
\begin{align}\nonumber
\min \Big\{f(\boldsymbol{y},\boldsymbol{\lambda})\!:&=\!\Big(\sum_{j=1}^{n}\lambda_{1j}\|y_1-y_j\|,\ldots,\sum_{j=1}^{n}\lambda_{nj}\|y_n-y_j\|\Big)|\ \boldsymbol{\lambda}\in \Lambda, \ \boldsymbol{y}\in \mathbb{R}^{n\times d}\Big\},
\end{align}
where $\Lambda=\{(\lambda_{ij})\in [0,1]^{n^2}: \sum_{j=1}^{n}\lambda_{ij}=1,\ \lambda_{ii}=0, \forall i\}$. As before, given a fixed opinion state $\boldsymbol{y}=\boldsymbol{x}(t)$, minimizing $f(\cdot)$ over $\boldsymbol{\lambda}\in \Lambda$ gives us $(\boldsymbol{\lambda}_t)_{ij}=1$ if $j=r(i)$, and $(\boldsymbol{\lambda}_t)_{ij}=0$ otherwise, where $r(\cdot)$ is defined with respect to the opinion profile $\boldsymbol{x}(t)$. Therefore, for a fixed state $\boldsymbol{y}=\boldsymbol{x}(t)$, the communication network is precisely captured by minimizing $f(\cdot)$ over $\boldsymbol{\lambda}\in \Lambda$. Now by fixing the network to this minimizer $\boldsymbol{\lambda}_t$, we get $f(\boldsymbol{x}(t),\boldsymbol{\lambda}_t)=(\|x_1(t)-x_{r(1)}(t)\|,\ldots,\|x_n(t)-x_{r(n)}(t)\|)$, which we next evaluate the effect of opinion updates on it.

Let $\boldsymbol{x}(t+1)$ be the updated opinion vector obtained from $\boldsymbol{x}(t)$, given the fixed network $\boldsymbol{\lambda}_t$. By definition of synchronous dynamics $x_i(t+1)=\mu x_i(t)+(1-\mu)x_{r(i)}(t), \forall i$,
\begin{align}\nonumber
\|x_i(t\!+\!1)\!-\!x_{r(i)}(t\!+\!1)\|&=\|\mu x_i(t)\!+\!(1\!-\!\mu)x_{r(i)}(t)\!-\!\big(\mu x_{r(i)}(t)\!+\!(1\!-\!\mu)x_{r(r(i))}(t)\big)\|\cr 
&=\|\mu (x_i(t)-x_{r(i)}(t))+(1-\mu)\big(x_{r(i)}(t)-x_{r(r(i))}(t)\big)\|\cr 
&\leq \mu \|x_i(t)-x_{r(i)}(t)\|+ (1-\mu) \|x_{r(i)}(t)-x_{r(r(i))}(t)\|\cr 
&\leq \mu \|x_i(t)-x_{r(i)}(t)\|+(1-\mu)\|x_{i}(t)-x_{r(i)}(t)\|\cr 
&=\|x_i(t)-x_{r(i)}(t)\|, 
\end{align}  
where the first inequality is due to the triangle inequality, and the last equality holds since by definition of $r(\cdot)$ we have $\|x_{r(i)}(t)-x_{r(r(i))}(t)\|\leq \|x_i(t)-x_{r(i)}(t)\|$. Therefore,   
\begin{align}\nonumber
f(\boldsymbol{x}(t\!+\!1),\boldsymbol{\lambda}_t)&\!=\!(\|x_i(t\!+\!1)\!-\!x_{r(i)}(t\!+\!1)\|,i\in[n])\leq (\|x_i(t)\!-\!x_{r(i)}(t)\|,i\in[n])\!=\!f(\boldsymbol{x}(t),\boldsymbol{\lambda}_t). 
\end{align} 
where the above inequality is component-wise. As a result, we have 
\begin{align}\nonumber
f(\boldsymbol{x}(t\!+\!1),\boldsymbol{\lambda}_{t+1})&=\min _{\boldsymbol{\lambda}\in \Lambda}f(\boldsymbol{x}(t\!+\!1),\boldsymbol{\lambda})\leq f(\boldsymbol{x}(t\!+\!1),\boldsymbol{\lambda}_t)\leq f(\boldsymbol{x}(t),\boldsymbol{\lambda}_t),
\end{align}
which shows that $V(\boldsymbol{y}):=\min_{\boldsymbol{\lambda}\in \Lambda}f(\boldsymbol{\lambda},\boldsymbol{y})=(\|y_i(t)-y_{r(i)}(t)\|,i\in[n])$ serves as a Lyapunov function for the synchronous nearest neighbor dynamics. 

To evaluate the convergence speed of the dynamics to an $\epsilon$-equilibrium point, let us consider the scalar function $\hat{V}(\boldsymbol{x})=\sum_{i=1}^{n}\|x_i\!-\!x_{r(i)}\|$. Let $C_t$ be a connected component of the communication graph at time $t$ which contains the longest edge $D_t:=\max_i \|x_i(t)-x_{r(i)}(t)\|$, and define $d_t$ to be the length of the shortest edge in $C_t$. A similar analysis as above shows that  
\begin{align}\label{eq:synchronous-drift}
\hat{V}(x(t))-\hat{V}(x(t&+1))\ge\sum_{i=1}^{n}\Big(\|x_i(t)-x_{r(i)}(t)\|-\|x_i(t+1)-x_{r(i)}(t+1)\|\Big)\cr 
&\ge\sum_{i\in C_t}\Big(\|x_i(t)-x_{r(i)}(t)\|-\|x_i(t+1)-x_{r(i)}(t+1)\|\Big)\cr
&\ge(1-\mu)\sum_{i\in C_t}\Big(\|x_i(t)-x_{r(i)}(t)\|-\|x_{r(i)}(t)-x_{r(r(i))}(t)\|\Big) \cr  
&\ge (1-\mu)\sum_{i\in P_t}\Big(\|x_i(t)-x_{r(i)}(t)\|-\|x_{r(i)}(t)-x_{r(r(i))}(t)\|\Big)\cr 
&=(1-\mu) (D_t-d_t), 
\end{align} 
where as before $r(\cdot)$ is associated with state $x(t)$, and $P_t$ is the unique directed path connecting the longest edge to the shortest edge in $C_t$. Note that the last equality in \eqref{eq:synchronous-drift} holds due to the telescopic sum over the nodes of $P_t$. 

On the other hand, by definition of $t_{\epsilon}$ we know that $D_t\ge \epsilon, \forall t<t_{\epsilon}$. We claim that for \emph{at most} $n\log_{|1-2\mu|} (\frac{\epsilon}{2D_0})$ time instances $t\in\{0,\ldots,t_{\epsilon}\}$ we can have $d_t\ge \frac{\epsilon}{2}$. This is because whenever an agent $i$ is an endpoint of the shortest edge in $C_t$, we get $r(r(i))=i$, and thus:
\begin{align}\nonumber
\|x_i(t+1)-x_{r_{t+1}(i)}(t+1)\|&\leq \|x_i(t+1)-x_{r(i)}(t+1)\|\cr 
&=\|\mu x_i(t)+(1-\mu)x_{r(i)}(t)-(\mu x_{r(i)}(t)-(1-\mu)x_i(t)) \|\cr 
&=|1-2\mu|\|x_i(t)-x_{r(i)}(t)\|. 
\end{align}
As a result the distance between agent $i$ and its nearest neighbor in the next time step will reduce by a factor of $|1-2\mu|\in(0,1)$, and as we saw earlier, for each $i\in[n]$ this distance can only decrease in the future iterations of the dynamics. This implies that agent $i\in[n]$ can be incident to the shortest edge in $C_t$ for at most $\log_{|1-2\mu|} (\frac{\epsilon}{2D_0})$ time instances before its distance to its closest neighbor shrinks below $\frac{\epsilon}{2}$. As there are in total $n$ agents, the claim follows. Therefore, for at least $t_\epsilon-n\log_{|1-2\mu|} (\frac{\epsilon}{2D_0})$ time instances we have $d_t< \frac{\epsilon}{2}$ and $D_t\ge \epsilon$. By \eqref{eq:synchronous-drift}, for such instances the function $\hat{V}(\cdot)$ must decrease by at least $D_t-d_t\ge \frac{\epsilon}{2}$. Since $\hat{V}(\cdot)$ is always nonnegative and $\hat{V}(\boldsymbol{x}(0))\leq nD_0$, we must have $(t_\epsilon-n\log_{|1-2\mu|} (\frac{\epsilon}{2D_0}))\frac{\epsilon}{2}<nD_0$, as desired.               
\end{proof}

\section{Applications to Game Theory}\label{sec:game}
In this section we provide a game-theoretic application of our successive optimization framework and show how it can be leveraged to establish convergence of natural best response dynamics (or its other variants) towards equilibrium points of the underlying game. For this purpose let us consider a \emph{Stackelberg} game with one leader (network designer) and $n$ followers (players). At each stage the leader decides about the network structure and the followers best respond to the leader's action. The leader's action is to choose a matrix $\boldsymbol{\lambda}:=(\lambda_{ij})\in \Lambda$, where $\lambda_{ij}$ indicates the amount by which the cost of player $i$ is influenced by the action of player $j$. After that each follower $i\in [n]$ incurs a cost of 
\begin{align}\nonumber
c_i(x_i,\boldsymbol{x}_{-i},\boldsymbol{\lambda})=\sum_{j=1}^{n}\lambda_{ij}\phi_i(x_i,x_j),
\end{align}
where $x_i\in X_i$ denotes the action taken by player $i$, $\boldsymbol{x}_{-i}$ denotes the actions of all players other than $i$, and $\phi_i(x_i,x_j)$ is a player-specific function which captures the coupling cost incurred by player $i$ due to the action of player $j$. For each player $i\in [n]$, we assume that $\phi_i(x_i,x_j)$ is a continuous and strictly convex function of its own action $x_i$, given fixed actions of all others (including the leader). Finally, in this game we assume that the leader's objective is to minimize the social cost defined by $c(\boldsymbol{x},\boldsymbol{\lambda})=\sum_{i=1}^{n}c_i (\boldsymbol{x},\boldsymbol{\lambda})$. This fully determines a Stackelberg game between the leader and the followers where the leader first determines the network structure and the followers best respond to it by playing a noncooperative game among themselves.

\begin{definition}
Let $X$ be a compact set, and consider a continuous function $f(x,y):X\times X\to \mathbb{R}$, where for each $x\in X$ the mapping $Z(x):=\arg\min_{y\in X}f(x,y)$ is singular value. $f(\cdot)$ is called min-symmetric, if $f(Z(x),Z(x))\leq f(x,Z(x)), \forall x\in X$.  
\end{definition}

\begin{remark}\label{rem:min-symmetric}
Quadratic functions of the form $f(\boldsymbol{x},\boldsymbol{y}):=\boldsymbol{x}^TQ_1\boldsymbol{x}+\boldsymbol{y}^TQ_2\boldsymbol{y}-2\boldsymbol{y}^TR\boldsymbol{x}$, where $Q_1, Q_2$ are positive definite matrices, and $R=R^T$ is a symmetric matrix, are a subclass of min-symmetric functions (see, e.g., \cite[Lemma 1]{roozbehani2008lyapunov}).
\end{remark}

\begin{theorem}\label{thm:stackelberg}
Let $\Lambda\subseteq [0,1]^{n\times n}$ be a polytope, and $X_i, i\in[n]$ be compact sets. Moreover, assume that for every $\boldsymbol{\lambda}\in \Lambda$, the function $\rho(\boldsymbol{x},\boldsymbol{y},\boldsymbol{\lambda}):=\sum_{i=1}^{n}c_i(y_i,\boldsymbol{x}_{-i},\boldsymbol{\lambda})$ is min-symmetric. Then every limit point of the best response dynamics generated by the leader and followers will be a Stackelberg equilibrium.    
\end{theorem}
\begin{proof}
Given a generic time instance $t$, let us first \emph{fix} leader's action to $\boldsymbol{\lambda}_t$ and analyze the game among followers. Note that for fixed $\boldsymbol{\lambda}_t$, the cost function of each player $i$ is continuous and strictly convex function of its own action $x_i$. Therefore, the unique best-response profile of all the followers at time $t+1$ is given by $\boldsymbol{x}(t+1)=\arg\min_{\boldsymbol{y}\in X} \rho(\boldsymbol{x}(t),\boldsymbol{y},\boldsymbol{\lambda}_t)$, where $X\!:=X_1\!\times\cdots\times X_n$ is a compact set. By min-symmetric property of $\rho(\boldsymbol{x},\boldsymbol{y},\boldsymbol{\lambda}_t)$, we can write,  
\begin{align}\label{eq:stakelberg}
c(\boldsymbol{x}(t+1),\boldsymbol{\lambda}_t)&=\rho(\boldsymbol{x}(t+1),\boldsymbol{x}(t+1),\boldsymbol{\lambda}_t)\cr 
&\leq \rho(\boldsymbol{x}(t),\boldsymbol{x}(t+1),\boldsymbol{\lambda}_t)\cr 
&=\min_{\boldsymbol{y}\in X} \rho(\boldsymbol{x}(t),\boldsymbol{y},\boldsymbol{\lambda}_t)\cr 
&\leq \rho(\boldsymbol{x}(t),\boldsymbol{x}(t),\boldsymbol{\lambda}_t)=c(\boldsymbol{x}(t),\boldsymbol{\lambda}_t). 
\end{align}
Finally, the leader's best action to the followers' actions $\boldsymbol{x}(t+1)$ is given by $\boldsymbol{\lambda}_{t+1}=\arg\min_{\boldsymbol{\lambda}\in\Lambda}c(\boldsymbol{x}(t+1),\boldsymbol{\lambda})$, which in view of \eqref{eq:stakelberg} shows that for any time $t$ we have, 
\begin{align}\label{eq:decresae-c}
c(\boldsymbol{x}(t+1),\boldsymbol{\lambda}_{t+1})\leq c(\boldsymbol{x}(t),\boldsymbol{\lambda}_t).
\end{align}

Next we note that $c(\boldsymbol{x},\boldsymbol{\lambda})$ is a linear function of $\boldsymbol{\lambda}$ which at each iteration is minimized over the polytope $\Lambda$. Therefore, the set of leader's best actions $\{\lambda_t\}_{t=0}^{\infty}$ is a finite set comprised of at most all the extreme points of $\Lambda$. Now let $(\boldsymbol{x}^*,\boldsymbol{\lambda}^*)$ be an arbitrary limit point of the sequence $\{(\boldsymbol{x}(t),\lambda_t)\}_{t=0}^{\infty}$ (which does exist since this sequence belongs to the compact set $X\times \Lambda$). As $\{\lambda_t\}_{t=0}^{\infty}$ is a finite set, this implies that there exists a subsequence $\{\boldsymbol{x}(t_k)\}_{k=0}^{\infty}$ which converges to $\boldsymbol{x}^*$ under the fixed network topology $\boldsymbol{\lambda}_{t_k}=\boldsymbol{\lambda}^*, \forall k\ge 0$. To derive a contradiction, assume that $\boldsymbol{\lambda}^*$ is not the best response of the leader to the followers' actions $\boldsymbol{x}^*$. This means that there exists a $\tilde{\boldsymbol{\lambda}}\neq \boldsymbol{\lambda}^*$ such that $c(\boldsymbol{x}^*,\tilde{\boldsymbol{\lambda}})<c(\boldsymbol{x}^*,\boldsymbol{\lambda}^*)$. By continuity of $c(\boldsymbol{x},\boldsymbol{\lambda})$ we have
\begin{align}\nonumber
\lim_{k\to \infty}c(\boldsymbol{x}(t_k),\tilde{\boldsymbol{\lambda}})=c(\boldsymbol{x}^*,\tilde{\boldsymbol{\lambda}})<c(\boldsymbol{x}^*,\boldsymbol{\lambda}^*)=\lim_{k\to \infty} c(\boldsymbol{x}(t_k),\boldsymbol{\lambda}_{t_k}).
\end{align} 
This means that there exists a sufficiently large integer $r>0$ such that $c(\boldsymbol{x}(t_r),\tilde{\boldsymbol{\lambda}})<c(\boldsymbol{x}(t_r),\boldsymbol{\lambda}_{t_r})$. But we already know that $c(\boldsymbol{x}(t_r),\boldsymbol{\lambda}_{t_r})=\min_{\boldsymbol{\lambda}\in\Lambda}c(\boldsymbol{x}(t_r),\boldsymbol{\lambda})\leq c(\boldsymbol{x}_{t_k},\tilde{\boldsymbol{\lambda}})$, which is in contrast with the former inequality. This shows that $\boldsymbol{\lambda}^*$ is the best response of the leader to followers' actions $\boldsymbol{x}^*$.

Similarly, given leader's action $\boldsymbol{\lambda}^*$, let us assume that $x^*$ is not a Nash equilibrium among the followers. Define $Z(\boldsymbol{x}):=\arg\min_{\boldsymbol{y}\in X}\rho(\boldsymbol{x},\boldsymbol{y},\boldsymbol{\lambda}^*)$ to be the best response function of the followers, given the fixed leader's action $\boldsymbol{\lambda}^*$. Note that $Z(x):X\to X$ is a continuous function due to uniform continuity of $\rho(\cdot,\cdot,\boldsymbol{\lambda}^*)$ over the compact set $X\times X$. Since $x^*$ is not a Nash equilibrium, this means that $\rho(\boldsymbol{x}^*,Z(\boldsymbol{x}^*),\boldsymbol{\lambda}^*)<\rho(\boldsymbol{x}^*,\boldsymbol{x}^*,\boldsymbol{\lambda}^*)$. As before, let $\{\boldsymbol{x}(t_k)\}_{k=0}^{\infty}$ be a subsequence converging to $\boldsymbol{x}^*$ under the fixed leader's action $\boldsymbol{\lambda}_{t_k}=\boldsymbol{\lambda}^*, \forall k\ge 0$. By continuity of $Z(x)$ and $\rho(\cdot,\cdot,\boldsymbol{\lambda}^*)$, we get 
\begin{align}\label{eq:final-delta}
\lim_{k\to \infty}\rho(\boldsymbol{x}(t_k),Z(\boldsymbol{x}(t_k)),\boldsymbol{\lambda}_{t_k})&=\lim_{k\to \infty}\rho(\boldsymbol{x}(t_k),Z(\boldsymbol{x}(t_k)),\boldsymbol{\lambda}^*)\cr 
&=\rho(\boldsymbol{x}^*,Z(\boldsymbol{x}^*),\boldsymbol{\lambda}^*)<\rho(\boldsymbol{x}^*,\boldsymbol{x}^*,\boldsymbol{\lambda}^*).
\end{align} 
Let $\delta:=\rho(\boldsymbol{x}^*,\boldsymbol{x}^*,\boldsymbol{\lambda}^*)-\rho(\boldsymbol{x}^*,Z(\boldsymbol{x}^*),\boldsymbol{\lambda}^*)>0$. From \eqref{eq:final-delta} and for sufficiently large integer $r$, we have $\rho(\boldsymbol{x}(t_r),Z(\boldsymbol{x}(t_r)),\boldsymbol{\lambda}_{t_r})<\rho(\boldsymbol{x}^*,\boldsymbol{x}^*,\boldsymbol{\lambda}^*)-\frac{\delta}{2}$. By definition we know that $\boldsymbol{x}(t_r+1)=Z(\boldsymbol{x}(t_r))$, thus by following the same argument as in \eqref{eq:stakelberg} and \eqref{eq:decresae-c}, 
\begin{align}\nonumber
c(\boldsymbol{x}(t_r+1),\boldsymbol{\lambda}_{t_r+1})&=\rho(\boldsymbol{x}(t_r+1),\boldsymbol{x}(t_r+1),\boldsymbol{\lambda}_{t_r+1})\cr 
&\leq \rho(\boldsymbol{x}(t_r+1),\boldsymbol{x}(t_r+1),\boldsymbol{\lambda}_{t_r})\cr 
&\leq  \rho(\boldsymbol{x}(t_r),Z(\boldsymbol{x}(t_r)),\boldsymbol{\lambda}_{t_r})\cr 
&<\rho(\boldsymbol{x}^*,\boldsymbol{x}^*,\boldsymbol{\lambda}^*)-\frac{\delta}{2}.
\end{align}
But from \eqref{eq:decresae-c} we know that $c(\boldsymbol{x}(t),\boldsymbol{\lambda}_{t})$ is a nonincreasing sequence so that for all $k> r$, we must have $c(\boldsymbol{x}(t_k),\boldsymbol{\lambda}_{t_k})\leq c(\boldsymbol{x}(t_r+1),\boldsymbol{\lambda}_{t_r+1})<\rho(\boldsymbol{x}^*,\boldsymbol{x}^*,\boldsymbol{\lambda}^*)-\frac{\delta}{2}$. Thus,
\begin{align}\nonumber
\rho(\boldsymbol{x}^*,\boldsymbol{x}^*,\boldsymbol{\lambda}^*)=c(\boldsymbol{x}^*,\boldsymbol{\lambda}^*)=\lim_{k\to \infty} c(\boldsymbol{x}(t_k),\boldsymbol{\lambda}_{t_k})\leq \rho(\boldsymbol{x}^*,\boldsymbol{x}^*,\boldsymbol{\lambda}^*)-\frac{\delta}{2}.
\end{align}
This contradiction shows that $\boldsymbol{x}^*$ is a Nash equilibrium among the followers, given the leader's action $\boldsymbol{\lambda}^*$. Therefore, the limit point $(\boldsymbol{x}^*,\boldsymbol{\lambda}^*)$ is a Stackelberg equilibrium.  
\end{proof}

\bigskip
\noindent
{\bf Example 1:} Consider a special case where $\phi_i(x,y)=(x-y)^2-\epsilon^2, \forall i\in[n]$, and the leader's best response polytope is given by the symmetric matrices with self-loops, i.e., $\Lambda:=\{\boldsymbol{\lambda}\in [0,1]^{n\times n}: \ \lambda_{ij}=\lambda_{ji},\forall i,j,\ \lambda_{ii}=1, \forall i\}$. In this case, we have 
\begin{align}\nonumber
\rho(\boldsymbol{x},\boldsymbol{y},\boldsymbol{\lambda})=\sum_{i=1}^{n}\sum_{j=1}^{n}\lambda_{ij}\big((y_i-x_j)^2-\epsilon^2\big)=\boldsymbol{y}^TQ\boldsymbol{y}+\boldsymbol{x}^TQ\boldsymbol{x}-2\boldsymbol{y}^T\boldsymbol{\lambda}\boldsymbol{y}-\epsilon^2(\boldsymbol{1}^T\boldsymbol{\lambda}\boldsymbol{1}),
\end{align} 
where again $\boldsymbol{\lambda}$ is a symmetric matrix, $Q=diag(\boldsymbol{\lambda}\boldsymbol{1})$ is positive definite, and hence by Remark \ref{rem:min-symmetric}, $\rho(\boldsymbol{x},\boldsymbol{y},\boldsymbol{\lambda})$ is min-symmetric. Therefore, the best response dynamics of the leader-followers converge to a Stackelberg equilibrium. Interestingly, one can easily check that the best response of followers are exactly governed by the homogeneous HK update rule which shows that the steady-state of the homogeneous HK model is simply a Stackelberg equilibrium of the above game. 

\smallskip
\noindent
{\bf Example 2:} Let us consider another special case where $\phi_i(x,y)=(x-y)^2, \forall i\in[n]$. Assume that at each instance the leader's objective is to keep the communication network among the followers (agents) connected, so that the followers can communicate with each other and eventually rendezvous at a consensus point. Perhaps one can think of agents as robots, and the leader as a bandwidth provider who aims to bring all the agents together with minimum bandwidth cost. In this case the leader's action set can be represent by the polytope 
\begin{align}\nonumber
\Lambda:=\{\boldsymbol{\lambda}\in [0,1]^{n\times n}: \ \sum_{i\in S, j\notin S}\lambda_{ij}\ge 1, \forall S\subset [n], \ \lambda_{ij}=\lambda_{ji} \forall i,j\}.
\end{align}
Note that the constraints $\sum_{i\in S, j\notin S}\lambda_{ij}\ge 1, \forall S\subset [n]$ guarantee that there is a communication path between every pair of agents. It is worth noting that the best response of the leader over the polytope $\Lambda$ can be solved efficiently in polynomial time despite the fact that $\Lambda$ contains exponentially many linear constraints. This can be done using Ellipsoid algorithm where the separation oracle can be implemented by solving at most $n^2$ network-flow problems. On the other hand, we have
\begin{align}\nonumber
\rho(\boldsymbol{x},\boldsymbol{y},\boldsymbol{\lambda})=\sum_{i=1}^{n}\sum_{j=1}^{n}\lambda_{ij}(y_i-x_j)^2=\boldsymbol{y}^TQ\boldsymbol{y}+\boldsymbol{x}^TQ\boldsymbol{x}-2\boldsymbol{y}^T\boldsymbol{\lambda}\boldsymbol{y},
\end{align} 
where $\boldsymbol{\lambda}$ is a symmetric matrix, and $Q=diag(\boldsymbol{\lambda}\boldsymbol{1})$ is a diagonal matrix with diagonal elements $Q_{ii}=\sum_j\lambda_{ij}\ge 1, \forall i\in[n]$. Again using Remark \ref{rem:min-symmetric} we conclude that $\rho(\boldsymbol{x},\boldsymbol{y},\boldsymbol{\lambda})$ is a min-symmetric function, which in view of Theorem \ref{thm:stackelberg} implies that the best response dynamics of the leader-followers converge to a Stackelberg equilibrium. Now let us find an explicit form for the best response dynamics. Given that at iteration $t$ the followers are at locations $\boldsymbol{x}(t)$, and the current min-cost network designed by the leader is $\boldsymbol{\lambda}_t$, at the next time instance the followers move to the average point of their neighbors given by $\boldsymbol{x}(t+1)=\arg\min_{\boldsymbol{y}}\rho(\boldsymbol{x}(t),\boldsymbol{y},\boldsymbol{\lambda}_t)=(Q_t^{-1}\boldsymbol{\lambda}_t)\boldsymbol{x}(t)$, or rewritten componentwise $x_i(t+1)=\sum_{j=1}^{n}\frac{(\boldsymbol{\lambda}_t)_{ij}}{(\boldsymbol{\lambda}_t\boldsymbol{1})_i}x_j(t), \forall i\in [n]$. Consequently, the leader looks at the current positions of the followers $\boldsymbol{x}(t+1)$, and assuming that making a communication link between two followers at positions $x_i(t+1)$ and $x_j(t+1)$ is proportional to their Euclidean distance $(x_i(t+1)-x_j(t+1))^2$, the leader decides what links to make (i.e., selects $\boldsymbol{\lambda}_{t+1}\in\Lambda$) so that the network remains connected while the network cost $c(\boldsymbol{x}(t+1),\boldsymbol{\lambda}_{t+1})$ is minimized. In other words, the leader solves the optimization problem $c(\boldsymbol{x}(t+1),\boldsymbol{\lambda}_{t+1})=\min_{\boldsymbol{\lambda}\in \Lambda} c(\boldsymbol{x}(t+1),\boldsymbol{\lambda})$. Clearly as the network remains connected at all the time instances while the followers take average over it, the convergent Stackelberg equilibrium must be a consensus point in which all the followers rendezvous at the same point.

\smallskip
As we close this section, we would like to note that in general, one can define different action polytopes $\Lambda$ for the leader and different influence functions $\phi_i(\cdot)$ for the followers. This allows us to recover various state-dependent network dynamics as the best response updates of leader-followers in a Stackelberg game, which by Theorem \ref{thm:stackelberg} are guaranteed to converge to a Stackelberg equilibrium, assuming the min-symmetric property. For instance, settings $\Lambda:=\{\boldsymbol{\lambda}\in [0,1]^{n\times n}: \ \sum_{j\neq i}\lambda_{ij}= 1, \forall i\in [n], \ \lambda_{ij}=\lambda_{ji} \forall i,j\}$ and quadratic influence functions $\phi_i(x,y)=(x-y)^2$ one can recover dynamics of a different variant of the nearest neighbor model considered in Section \ref{sec:nearest-neighbor}, and establish their convergence to an equilibrium point. In conclusion, game theory can provide alternative shortcuts for the analysis of dynamic networks of multi-agent systems. An interesting fact about this approach is that many of the existing results developed in the game-theory literature such as \cite{rosen1965existence} do not require stringent homogeneity or symmetric assumptions among the players. This provides ample room for analysis of heterogeneous multi-agent network dynamics by bridging them to game-theoretic settings. 

\section{Conclusions and Further Discusssions}\label{sec:conclusion}
In this paper, we have studied Lyapunov stability and convergence of multi-agent network systems with state-dependent switching dynamics. By incorporating the network structure into our framework as a new variable, we have shown that how the evolution of multi-agent network dynamics can be viewed as trajectories of successive optimization problems. Leveraging this framework, we were able to establish Lyapunov stability and convergence of several well-known models from social science, and extended our results to scenarios with asymmetric communication structures. In particular, we showed how these results can be viewed from a game-theoretic perspective, so that convergence of multi-agent network systems can be interpreted as selfish behavior of players in a well-designed Stackelberg game. 

We believe that this work opens many new directions for stability analysis of multi-agent network systems with rich state-dependent switching dynamics. In the following we have listed a few of such directions:

1) Often in successive optimization framework one adds smooth proximal functions to the original objective function (or approximates it with a smooth upper bound) in order to facilitates the minimization sub-problems \cite{attouch2010proximal}. Incorporating this idea into our framework will give us perturbed multi-agent network dynamics whose update rules are affected by the proximal term (while their convergence to a stationary point is still guaranteed). For instance, one can add a proximal term to the objective function $f(\boldsymbol{y},\boldsymbol{\lambda})$ of the homogeneous HK model in terms of Bregman distance of a strongly convex function and recover noisy versions of HK dynamics.

2) In Section \ref{sec:HK}, for a fixed network $\boldsymbol{\lambda}^*$ with associate Laplacian $\mathcal{L}^*$, we only used the quadratic Lyapunov function $\Phi(\boldsymbol{x},\boldsymbol{\lambda}^*)=\boldsymbol{x}^T\mathcal{L}^*\boldsymbol{x}$ into our BCD framework. While this Lyapunov function seems quite suitable when the network structure is symmetric (undirected), it cannot serve as a Lyapunov function for general asymmetric (directed) networks. However, for a general \emph{fixed} directed network with adjacency matrix $\boldsymbol{\lambda}^*$, one can choose an alternative Lyapunov function given by $\Phi(\boldsymbol{x},\boldsymbol{\pi}^*)=\boldsymbol{x}^T(diag(\boldsymbol{\pi}^*)-\boldsymbol{\pi}^*(\boldsymbol{\pi}^*)^T)\boldsymbol{x}=\sum_{i,j}\pi^*_i\pi^*_j(x_i-x_j)^2$, where $\boldsymbol{\pi}^*$ is the Perron-left eigenvector of the transition matrix of a simple random walk over $\boldsymbol{\lambda}^*$, i.e., the Peron-left eigenvector of $A=(diag(\boldsymbol{\lambda}^*\boldsymbol{1}))^{-1}\boldsymbol{\lambda}^*$ so that $(\boldsymbol{\pi}^*)^TA=(\boldsymbol{\pi}^*)^T$. In this case, one can replace the network variable matrix $\boldsymbol{\lambda}^*$ by its Perron-left eigenvector $\boldsymbol{\pi}^*$ in the BCD framework. Therefore, as far as minimizing $\boldsymbol{\pi}$ with respect to its constraint set matches the left-Perron eigenvector of the actual network dynamics, one can use Proposition \ref{prop:framework} with $\Phi(\boldsymbol{x},\boldsymbol{\pi}^*)=\boldsymbol{x}^T(diag(\boldsymbol{\pi}^*)-\boldsymbol{\pi}^*(\boldsymbol{\pi}^*)^T)\boldsymbol{x}$ to show Lyapunov stability of joint state-network dynamics under more general asymmetric (directed) environment. Therefore, it is interesting to study stability of state-dependent multi-agent networks under highly asymmetric setting using $\boldsymbol{\pi}$ variables, rather than simply the edge variables $\boldsymbol{\lambda}=(\lambda_{ij})$.      

3) In Section \ref{sec:game}, we used a Stackelberg game where the followers' cost functions are convex (i.e. the noncooperative game among followers is a convex game \cite{rosen1965existence}) and the leader's objective function is to minimize the social cost. These conditions can be generalized to other settings. For instance, one can consider a Stackelberg network game where the game among followers is a potential game \cite{monderer1996potential} whose potential function can serve as the Lyapunov function $\Phi(\boldsymbol{x},\boldsymbol{\lambda}^*)$ into our BCD framework, given a fixed leader's action $\boldsymbol{\lambda}^*$.

Finally, in this paper we mainly focused on the stability and convergence analysis of the underlying dynamics. Therefore, an interesting direction here is to characterize the structure of the equilibrium points in more detail. 

\bibliographystyle{IEEEtran}
\bibliography{thesisrefs}

\begin{thebibliography}{10}
\providecommand{\url}[1]{#1}
\csname url@samestyle\endcsname
\providecommand{\newblock}{\relax}
\providecommand{\bibinfo}[2]{#2}
\providecommand{\BIBentrySTDinterwordspacing}{\spaceskip=0pt\relax}
\providecommand{\BIBentryALTinterwordstretchfactor}{4}
\providecommand{\BIBentryALTinterwordspacing}{\spaceskip=\fontdimen2\font plus
\BIBentryALTinterwordstretchfactor\fontdimen3\font minus
  \fontdimen4\font\relax}
\providecommand{\BIBforeignlanguage}[2]{{%
\expandafter\ifx\csname l@#1\endcsname\relax
\typeout{** WARNING: IEEEtran.bst: No hyphenation pattern has been}%
\typeout{** loaded for the language `#1'. Using the pattern for}%
\typeout{** the default language instead.}%
\else
\language=\csname l@#1\endcsname
\fi
#2}}
\providecommand{\BIBdecl}{\relax}
\BIBdecl

\bibitem{degroot1974reaching}
M.~H. DeGroot, ``Reaching a consensus,'' \emph{Journal of the American
  Statistical Association}, vol.~69, no. 345, pp. 118--121, 1974.

\bibitem{friedkin1997social}
N.~E. Friedkin and E.~C. Johnsen, ``Social positions in influence networks,''
  \emph{Social Networks}, vol.~19, no.~3, pp. 209--222, 1997.

\bibitem{nedic2018network}
A.~Nedi{\'c}, A.~Olshevsky, and M.~G. Rabbat, ``Network topology and
  communication-computation tradeoffs in decentralized optimization,''
  \emph{Proceedings of the IEEE}, vol. 106, no.~5, pp. 953--976, 2018.

\bibitem{AVP-RT:17}
A.~V. Proskurnikov and R.~Tempo, ``A tutorial on modeling and analysis of
  dynamic social networks. {Part I},'' \emph{Annual Reviews in Control},
  vol.~43, pp. 65--79, 2017.

\bibitem{olshevsky2009convergence}
A.~Olshevsky and J.~N. Tsitsiklis, ``Convergence speed in distributed consensus
  and averaging,'' \emph{SIAM Journal on Control and Optimization}, vol.~48,
  no.~1, pp. 33--55, 2009.

\bibitem{nedic2009distributed}
A.~Nedic, A.~Olshevsky, A.~Ozdaglar, and J.~N. Tsitsiklis, ``On distributed
  averaging algorithms and quantization effects,'' \emph{IEEE Transactions on
  Automatic Control}, vol.~54, no.~11, pp. 2506--2517, 2009.

\bibitem{nedic2015distributed}
A.~Nedi{\'c} and A.~Olshevsky, ``Distributed optimization over time-varying
  directed graphs,'' \emph{IEEE Transactions on Automatic Control}, vol.~60,
  no.~3, pp. 601--615, 2015.

\bibitem{bacsar2016convergence}
T.~Ba{\c{s}}ar, S.~R. Etesami, and A.~Olshevsky, ``Convergence time of
  quantized {M}etropolis consensus over time-varying networks,'' \emph{IEEE
  Transactions on Automatic Control}, vol.~61, no.~12, pp. 4048--4054, 2016.

\bibitem{zhu2011convergence}
M.~Zhu and S.~Mart{\'\i}nez, ``On the convergence time of asynchronous
  distributed quantized averaging algorithms,'' \emph{IEEE Transactions on
  Automatic Control}, vol.~56, no.~2, pp. 386--390, 2011.

\bibitem{tatarenko2017non}
T.~Tatarenko and B.~Touri, ``Non-convex distributed optimization,'' \emph{IEEE
  Transactions on Automatic Control}, vol.~62, no.~8, pp. 3744--3757, 2017.

\bibitem{olfati2004consensus}
R.~Olfati-Saber and R.~M. Murray, ``Consensus problems in networks of agents
  with switching topology and time-delays,'' \emph{IEEE Transactions on
  Automatic Control}, vol.~49, no.~9, pp. 1520--1533, 2004.

\bibitem{hendrickx2013convergence}
J.~M. Hendrickx and J.~N. Tsitsiklis, ``Convergence of type-symmetric and
  cut-balanced consensus seeking systems,'' \emph{IEEE Transactions on
  Automatic Control}, vol.~58, no.~1, pp. 214--218, 2013.

\bibitem{jadbabaie2003coordination}
A.~Jadbabaie, J.~Lin, and A.~S. Morse, ``Coordination of groups of mobile
  autonomous agents using nearest neighbor rules,'' \emph{IEEE Transactions on
  automatic control}, vol.~48, no.~6, pp. 988--1001, 2003.

\bibitem{touri2012product}
B.~Touri, \emph{Product of random stochastic matrices and distributed
  averaging}.\hskip 1em plus 0.5em minus 0.4em\relax Springer Science \&
  Business Media, 2012.

\bibitem{sonin2008decomposition}
I.~M. Sonin \emph{et~al.}, ``The decomposition-separation theorem for finite
  nonhomogeneous {M}arkov chains and related problems,'' in \emph{Markov
  Processes and Related Topics: A Festschrift for Thomas G. Kurtz}.\hskip 1em
  plus 0.5em minus 0.4em\relax Institute of Mathematical Statistics, 2008, pp.
  1--15.

\bibitem{touri2011existence}
B.~Touri and A.~Nedi{\'c}, ``On existence of a quadratic comparison function
  for random weighted averaging dynamics and its implications,'' in \emph{50th
  IEEE Conference on Decision and Control and European Control Conference
  (CDC-ECC)}.\hskip 1em plus 0.5em minus 0.4em\relax IEEE, 2011, pp.
  3806--3811.

\bibitem{etesami2013termination}
S.~R. Etesami, T.~Ba{\c{s}}ar, A.~Nedi{\'c}, and B.~Touri, ``Termination time
  of multidimensional {H}egselmann-{K}rause opinion dynamics,'' in
  \emph{American Control Conference (ACC)}.\hskip 1em plus 0.5em minus
  0.4em\relax IEEE, 2013, pp. 1255--1260.

\bibitem{razaviyayn2013unified}
M.~Razaviyayn, M.~Hong, and Z.-Q. Luo, ``A unified convergence analysis of
  block successive minimization methods for nonsmooth optimization,''
  \emph{SIAM Journal on Optimization}, vol.~23, no.~2, pp. 1126--1153, 2013.

\bibitem{howson1975new}
H.~Howson and N.~Sancho, ``A new algorithm for the solution of multi-state
  dynamic programming problems,'' \emph{Mathematical Programming}, vol.~8,
  no.~1, pp. 104--116, 1975.

\bibitem{hartigan1979algorithm}
J.~Hartigan and M.~Wong, ``A k-means clustering algorithm,'' \emph{Journal of
  the Royal Statistical Society. Series C (Applied Statistics)}, vol.~28,
  no.~1, pp. 100--108, 1979.

\bibitem{hegselmann2002opinion}
R.~Hegselmann and U.~Krause, ``Opinion dynamics and bounded confidence models,
  analysis and simulation,'' \emph{Journal of Artificial Societies and Social
  Simulation}, vol.~5, no.~3, pp. 1--33, 2002.

\bibitem{friedkin1999social}
N.~E. Friedkin and E.~C. Johnsen, ``Social influence networks and opinion
  change,'' \emph{Advances in Group Processes}, vol.~16, no.~1, pp. 1--29,
  1999.

\bibitem{friedkin2011social}
------, \emph{Social influence network theory: {A} sociological examination of
  small group dynamics}.\hskip 1em plus 0.5em minus 0.4em\relax Cambridge
  University Press, 2011, vol.~33.

\bibitem{lorenz2007continuous}
J.~Lorenz, ``Continuous opinion dynamics under bounded confidence: A survey,''
  \emph{International Journal of Modern Physics C}, vol.~18, no.~12, pp.
  1819--1838, 2007.

\bibitem{proskurnikov2018tutorial}
A.~V. Proskurnikov and R.~Tempo, ``A tutorial on modeling and analysis of
  dynamic social networks. {P}art {II},'' \emph{Annual Reviews in Control},
  vol.~45, pp. 166--190, 2018.

\bibitem{Bullo}
F.~Bullo, J.~Cortes, and S.~Martinez, \emph{Distributed control of robotic
  networks: {A} mathematical approach to motion coordination algorithms}.\hskip
  1em plus 0.5em minus 0.4em\relax Princeton University Press, 2009, vol.~27.

\bibitem{chazelle2011total}
B.~Chazelle, ``The total s-energy of a multiagent system,'' \emph{SIAM Journal
  on Control and Optimization}, vol.~49, no.~4, pp. 1680--1706, 2011.

\bibitem{dong2016dynamics}
Y.~Dong, X.~Chen, H.~Liang, and C.-C. Li, ``Dynamics of linguistic opinion
  formation in bounded confidence model,'' \emph{Information Fusion}, vol.~32,
  pp. 52--61, 2016.

\bibitem{ye2018evolution}
M.~B. Ye, J.~Liu, B.~D. Anderson, C.~B. Yu, and T.~Basar, ``Evolution of social
  power in social networks with dynamic topology,'' \emph{IEEE Transactions on
  Automatic Control}, vol.~63, pp. 3793 -- 3808, 2018.

\bibitem{wai2015identifying}
H.-T. Wai, A.~Scaglione, and A.~Leshem, ``Identifying trust in social networks
  with stubborn agents, with application to market decisions,'' in \emph{53rd
  Annual Allerton Conference on Communication, Control, and Computing}.\hskip
  1em plus 0.5em minus 0.4em\relax IEEE, 2015, pp. 747--754.

\bibitem{lorenzt}
J.~Lorenz, ``Repeated averaging and bounded-confidence, modeling, analysis and
  simulation of continuous opinion dynamics,'' Ph.D. dissertation, University
  of Bremen, 2007.

\bibitem{mirtabatabaei2012opinion}
A.~Mirtabatabaei and F.~Bullo, ``Opinion dynamics in heterogeneous networks:
  {C}onvergence conjectures and theorems,'' \emph{SIAM Journal on Control and
  Optimization}, vol.~50, no.~5, pp. 2763--2785, 2012.

\bibitem{lorenz2010heterogeneous}
J.~Lorenz, ``Heterogeneous bounds of confidence: Meet, discuss and find
  consensus!'' \emph{Complexity}, vol.~15, no.~4, pp. 43--52, 2010.

\bibitem{julient}
J.~M. Hendrickx, ``Graphs and networks for the analysis of autonomous agent
  systems,'' \emph{Ph.D. Thesis, Universite Catholique de Louvain}, 2011.

\bibitem{chazelle2017inertial}
B.~Chazelle and C.~Wang, ``Inertial {H}egselmann-{K}rause systems,'' \emph{IEEE
  Transactions on Automatic Control}, vol.~62, no.~8, pp. 3905--3913, 2017.

\bibitem{pineda2013noisy}
M.~Pineda, R.~Toral, and E.~Hern{\'a}ndez-Garc{\'\i}a, ``The noisy
  {H}egselmann-{K}rause model for opinion dynamics,'' \emph{The European
  Physical Journal B}, vol.~86, no.~12, pp. 1--10, 2013.

\bibitem{blondel20072r}
V.~D. Blondel, J.~M. Hendrickx, and J.~N. Tsitsiklis, ``On the {$2R$}
  conjecture for multi-agent systems,'' in \emph{European Control Conference
  (ECC)}.\hskip 1em plus 0.5em minus 0.4em\relax IEEE, 2007, pp. 874--881.

\bibitem{hendrickx2013symmetric}
J.~M. Hendrickx and A.~Olshevsky, ``On symmetric continuum opinion dynamics,''
  \emph{SIAM Journal on Control and Optimization}, vol.~54, no.~5, pp.
  2893--2918, 2016.

\bibitem{bhattacharyya2013convergence}
A.~Bhattacharyya, M.~Braverman, B.~Chazelle, and H.~L. Nguyen, ``On the
  convergence of the {H}egselmann-{K}rause system,'' in \emph{Proceedings of
  the 4th Conference on Innovations in Theoretical Computer Science}.\hskip 1em
  plus 0.5em minus 0.4em\relax ACM, 2013, pp. 61--66.

\bibitem{roozbehani2008lyapunov}
M.~Roozbehani, A.~Megretski, and E.~Frazzoli, ``Lyapunov analysis of
  quadratically symmetric neighborhood consensus algorithms,'' in \emph{Proc.
  47th IEEE Conference on Decision and Control (CDC)}.\hskip 1em plus 0.5em
  minus 0.4em\relax IEEE, 2008, pp. 2252--2257.

\bibitem{zhang2012lyapunov}
H.~Zhang, F.~L. Lewis, and Z.~Qu, ``Lyapunov, adaptive, and optimal design
  techniques for cooperative systems on directed communication graphs,''
  \emph{IEEE Transactions on Industrial Electronics}, vol.~59, no.~7, pp.
  3026--3041, 2012.

\bibitem{alex}
A.~Olshevsky and J.~Tsitsiklis, ``On the nonexistence of quadratic {L}yapunov
  functions for consensus algorithms,'' \emph{IEEE Transactions on Automatic
  Control}, vol.~53, pp. 2642--2645, 2008.

\bibitem{etesami2015game}
S.~R. Etesami and T.~Ba\c{s}ar, ``Game-theoretic analysis of the
  {H}egselmann-{K}rause model for opinion dynamics in finite dimensions,''
  \emph{IEEE Transactions on Automatic Control}, vol.~60, no.~7, pp.
  1886--1897, 2015.

\bibitem{su2017partial}
W.~Su, Y.~Gu, S.~Wang, and Y.~Yu, ``Partial convergence of heterogeneous
  {H}egselmann-{K}rause opinion dynamics,'' \emph{Science China Technological
  Sciences}, vol.~60, no.~9, pp. 1433--1438, 2017.

\bibitem{rosen1965existence}
J.~B. Rosen, ``Existence and uniqueness of equilibrium points for concave
  $n$-person games,'' \emph{Econometrica: Journal of the Econometric Society},
  pp. 520--534, 1965.

\bibitem{attouch2010proximal}
H.~Attouch, J.~Bolte, P.~Redont, and A.~Soubeyran, ``Proximal alternating
  minimization and projection methods for nonconvex problems: {A}n approach
  based on the kurdyka-{L}ojasiewicz inequality,'' \emph{Mathematics of
  Operations Research}, vol.~35, no.~2, pp. 438--457, 2010.

\bibitem{monderer1996potential}
D.~Monderer and L.~S. Shapley, ``Potential games,'' \emph{Games and Economic
  Behavior}, vol.~14, no.~1, pp. 124--143, 1996.

\end{thebibliography}

\newpage
\section{Appendix I}

\smallskip
\begin{lemma}\label{lemm:eulerian-graph}
Consider the restricted HK model and choose $\delta<\frac{1}{n^2}$. Let $t_{\delta}$ be a time instance for which $\sum_{\tau=t_{\delta}}^{\infty}\sum_{i=1}^{n}(x_i(\tau+1)-x_i(\tau))^2<\delta^2$. Then for any $t\ge t_{\delta}$ the communication network $\boldsymbol{\lambda}_t$ remains unchanged.  
\end{lemma}
\begin{proof}
By contradiction, let us assume that the lemma does not hold and $t\ge t_{\delta}$ be a time instance for which the communication network changes, i.e., $\boldsymbol{\lambda}_t\neq \boldsymbol{\lambda}_{t+1}$. We construct an undirected graph $\mathcal{B}$, so-called \emph{balanced graph}, as follows: 
\begin{itemize}
\item Nodes of $\mathcal{B}$ are the points $\{x_1(t+1),x_2(t+1),\ldots,x_n(t+1)\}$ which are positioned on the real axis. For simplicity, we refer to these nodes by their indicies such that $i$ represents the node positioned at $x_i(t+1)$. Note that in $\mathcal{B}$ both labels and geometric positions of the vertices are important. 
\item The edges of $\mathcal{B}$ are partitioned into two groups: \emph{solid} and \emph{dashed}. There is a \emph{solid} edge between two nodes $i$ and $j$ if and only if $\{i,j\}\in \mathcal{E}, |x_i(t)-x_j(t)|>1$, and $|x_i(t+1)-x_j(t+1)|\leq 1$. Thus a solid edge between nodes $i$ and $j$ in $\mathcal{B}$ indicates that agents $i$ and $j$ were not each others' neighbors at time $t$ (i.e., $(\boldsymbol{\lambda}_t)_{ij}=0$), while they become neighbor at time $t+1$ (i.e., $(\boldsymbol{\lambda}_{t+1})_{ij}=1$). There is a \emph{dashed} edge between $i$ and $j$ if and only if $\{i,j\}\in \mathcal{E}, |x_i(t)-x_j(t)|\leq 1$, and $|x_i(t+1)-x_j(t+1)|> 1$. Thus a dashed edge in $\mathcal{B}$ shows that agents $i$ and $j$ were each others' neighbors at time $t$, but they separate at time $t+1$.    
\end{itemize} 


Next we show that if $t\ge t_{\delta}$ is a time instance for which the communication network changes, then the following three facts must hold: {\bf Fact 1)} For every two nodes $i$ and $j$ which are connected by an edge in $\mathcal{B}$, we have $1-2\delta<|x_i(t+1)-x_j(t+1)|<1+2\delta$. This is an immediate consequence of the edge definition given above together with the fact that no agent can move by a distance more than $\delta$ from time $t$ to $t+1$.  {\bf Fact 2)} The geometric distance between \emph{any} two nodes which are connected by a dashed edge is strictly greater than the distance between \emph{any} other two nodes which are connected by a solid edge. This is because the former is strictly greater than 1, while the later is at most 1. {\bf Fact 3)} Let $s_L(i)$ and $d_L(i)$ denote, respectively, the number of solid and dashed edges which are incident to node $i$ from left hand side. Similarly, let $s_R(i)$ and $d_R(i)$ be the number of solid/dashed edges which are incident to node $i$ from right hand side. Then a straightforward calculation as in \cite[Equation 2]{chazelle2017inertial} shows that for sufficiently small $\delta\leq \frac{1}{n^2}$, we must have $s_L(i)-d_L(i)=s_R(i)-d_R(i), \forall i$. In other words, every node in $\mathcal{B}$ must be \emph{balanced} with respect to the `effective' amount of change in its neighbors from time $t$ to $t+1$. Otherwise moving from time $t+1$ to $t+2$, an unbalanced agent will move by a distance of at least $\delta$, contradicting the fact that $t\ge t_{\delta}$.\footnote{Intuitively, $s_L(i)-d_L(i)$ measures the `effective' change due to change of left-neighbors of agent $i$ from time $t$ to $t+1$. This must be equal to $s_R(i)-d_R(i)$ which is the effective change in the right-neighbors of $i$. Otherwise, agent $i$ will move by a large distance at the next time instance $t+2$.}

To derive a contradiction, consider the time $t\ge t_{\delta}$ for which a switch in the communication network occurs. This means that the associated balance graph $\mathcal{B}$ has at least one edge. Henceforth, and by some abuse of notation, we denote a nontrivial connected component of this balanced graph by  $\mathcal{B}$ (which is a connected graph with at least one edge and no isolated vertex). Since Fact 1 holds, the vertices of $\mathcal{B}$ can be covered by the union of some \emph{disjoint} intervals $\{I_{\ell}=(a_\ell,b_{\ell}): \ell=1,\ldots,k\}$, where $a_{1}<b_1<a_2<b_2,\ldots<b_{k}$. Moreover, each of these intervals has a very small length (at most $b_{\ell}-a_{\ell}<4n\delta$), and any edge of $\mathcal{B}$ has its end points in two consecutive intervals $I_{\ell},I_{\ell+1}$ (see Figure \ref{Fig:ilands} for an illustration).\footnote{Note that this also implies that every two consecutive intervals are apart from each other by a large distance of at least $1-2\delta-8n\delta$.}

\begin{figure}
  \begin{center}
    \includegraphics[totalheight=.27\textheight,
width=.4\textwidth,viewport=350 0 950 600]{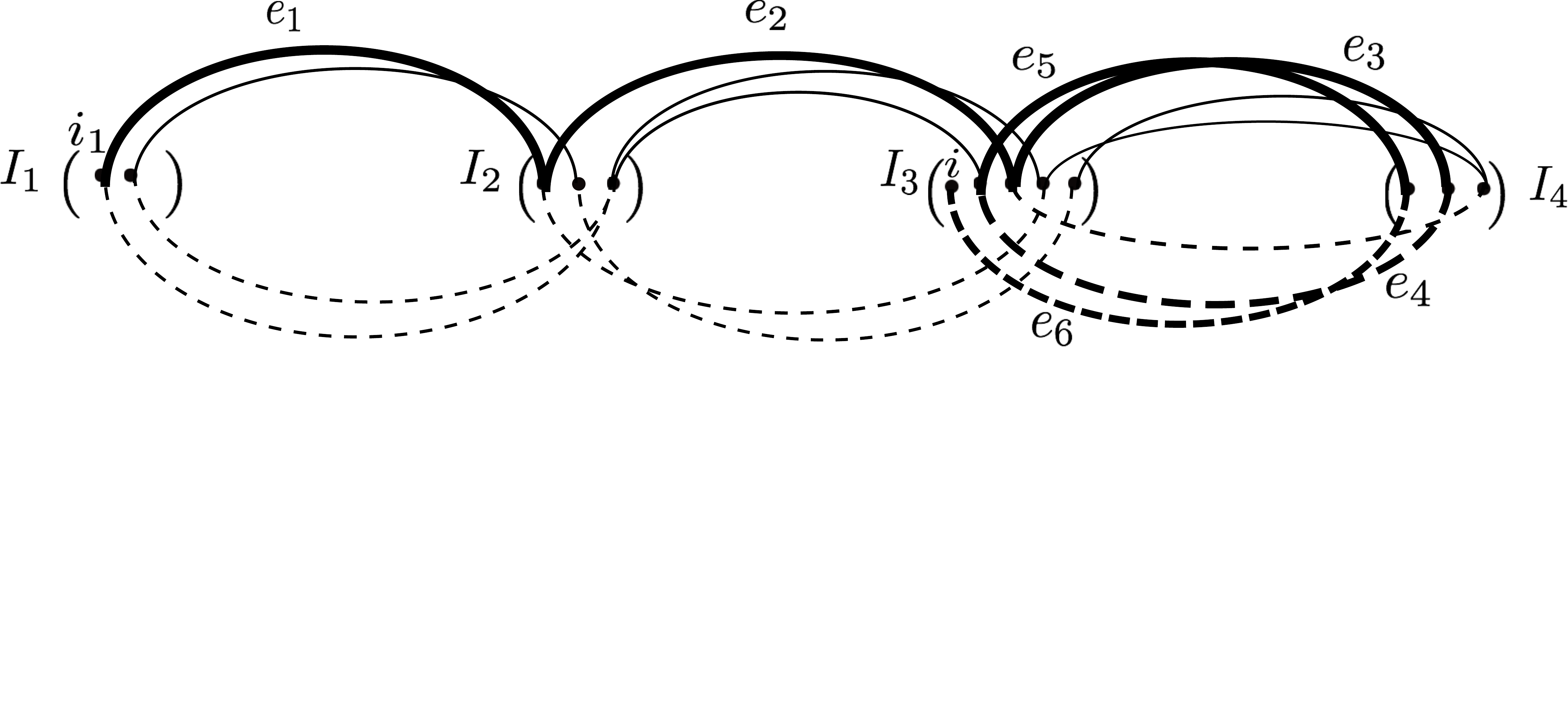}
  \end{center}\vspace{-3cm}
  \caption{\footnotesize{An illustration of the balanced graph $\mathcal{B}$ with solid/dashed edges and covering intervals $I_1,I_2,I_3$, and $I_4$. Here a maximal walk $W=(e_1,e_2,e_3,e_4,e_5,e_6)$ with initial node $i_1$ and endpoint $i$ is represented by thick edges. Note that for node $i_1$ we have $s_L(i_1)=d_L(i_1)=0$ and $s_R(i_1)=d_R(i_1)=1$. Therefore, $s_L(i_1)-d_L(i_1)=s_R(i_1)-d_R(i_1)$, which means that node $i_1$ is balanced. However for node $i$ we have $s_L(i)=d_L(i)=s_R(i)=0$, and $d_R(i)=1$. Thus $s_L(i)-d_L(i)=0\neq -1=s_R(i)-d_R(i)$, which means that node $i$ is not balanced. \vspace{-0.2cm}}}\label{Fig:ilands}
\end{figure} 

Next let us consider the most left agent in the connected balance graph $\mathcal{B}$ and call it $i_1$ (breaking ties arbitrarily). By Fact 3 every node in $\mathcal{B}$ (and in particular $i_1$) is balanced. Since $i_1$ does not have any edge incident to it from left hand side, it must have at least one solid edge incident to it from the right hand side. Now starting from node $i_1$, let us consider a maximal \emph{alternating walk} $W$ which sequentially traverses over the edges of $\mathcal{B}$ from left to right using \emph{solid} edges, and from right to left using \emph{dashed} edges (Figure \ref{Fig:ilands}). As we argued above, $W$ is nonempty and contains at least one solid edge. We claim that $W$ must be a \emph{path}, meaning that no vertex can be visited more than once by $W$. Otherwise, consider the first time when $W$ visits a vertex for the second time, which results to an alternating cycle $C$. As each vertex of $C$ lies in one of the disjoint intervals $I_{\ell}$, and each edge of $C$ has its endpoints in two consecutive intervals, this implies that $C$ must have the same number of solid and dashed edges. This in view of Fact 2 shows that the sum of lengths of solid edges in $C$ is strictly less than the sum of lengths of dashed edges in $C$. However, we know that for any alternating cycle the sum of lengths of solid edges (the total movement from left to right) must be equal to the sum of lengths of dashed edges (the total movement from right to left). This contradiction shows that $W$ must be a path. 

Finally, let us denote the other endpoint of path $W$ by $i$. Either node $i$ is reached through the path $W$ using a dashed edge (and hence from right to left), in which case $d_R(i)\ge 1$ and $s_R(i)=d_L(i)=0$ (otherwise $W$ can be extended and is not maximal). Or node $i$ is reached using a solid edge (and hence from left to right), in which case $s_L(i)\ge 1$ and $s_R(i)=d_L(i)=0$. Now an easy calculation shows that in either case $s_L(i)-d_L(i)\neq s_R(i)-d_R(i)$, meaning that node $i$ is not a balanced node. This is in contradiction with Fact 3, which completes the proof.   
\end{proof}

\end{document}